\documentclass[12pt,leqno]{amsart}

\usepackage{amsmath,amssymb,amscd,amsthm}
\usepackage{latexsym}
\usepackage{graphicx}

\newtheorem{theorem}{Theorem}
\newtheorem{lemma}{Lemma}
\newtheorem{proposition}{Proposition}

\newtheorem{corollary}{Corollary}

\newtheorem{claim}{Claim}

\newcommand{\f}[2]{\frac{#1}{#2}}
\newcommand{\dpr}[2]{\langle #1,#2 \rangle}

\newcommand{\al}{\alpha}

\newcommand{\De}{\Delta}

\newcommand{\La}{\Lambda}

\newcommand{\rone}{\mathbf R^1}
\newcommand{\rtwo}{\mathbf R^2}

\newcommand{\cf}{\mathcal F}

\newcommand{\p}{\partial}

\newcommand{\sgn}{\mathrm{sgn}}

\newcommand{\beq}{\begin{equation}}
\newcommand{\eeq}{\end{equation}}
\newcommand{\beqna}{\begin{eqnarray*}}
\newcommand{\eeqna}{\end{eqnarray*}}
\newcommand{\beqn}{\begin{equation*}}
\newcommand{\eeqn}{\end{equation*}}
\newcommand{\bp}{\begin{proof}}
\newcommand{\ep}{\end{proof}}
\newcommand{\bprop}{\begin{proposition}}
\newcommand{\eprop}{\end{proposition}}
\newcommand{\bt}{\begin{theorem}}
\newcommand{\et}{\end{theorem}}
\newcommand{\bex}{\begin{Example}}
\newcommand{\eex}{\end{Example}}
\newcommand{\bc}{\begin{corollary}}
\newcommand{\ec}{\end{corollary}}
\newcommand{\bcl}{\begin{claim}}
\newcommand{\ecl}{\end{claim}}
\newcommand{\bl}{\begin{lemma}}
\newcommand{\el}{\end{lemma}}

\usepackage{color}

\begin{document}

\title
[Mass  in Mass system of granular chains]
{Traveling waves for  the Mass in Mass model of granular chains}

\author{Panayotis G. Kevrekidis}

\author{Atanas G. Stefanov}

\author{Haitao Xu}

\address{Panayotis G. Kevrekidis\\
Lederle Graduate Research Tower\\
Department of Mathematics and  Statistics\\
University of Massachusetts\\
Amherst, MA 01003 \\
and Center for Nonlinear Studies and Theoretical Division \\
Los Alamos National Laboratory \\
Los Alamos, NM 87544}

\address{Atanas G.  Stefanov\\
Department of Mathematics \\
University of Kansas\\
1460 Jayhawk Blvd\\ Lawrence, KS 66045--7523}
\email{stefanov@ku.edu}

\address{Haitao Xu\\
Lederle Graduate Research Tower\\
Department of Mathematics and  Statistics\\
University of Massachusetts\\
Amherst, MA 01003}

\email{kevrekid@math.umass.edu}

\thanks{Stefanov's research is supported in part by
 NSF-DMS 1313107.  Kevrekidis
acknowledges support from the National Science Foundation under
grant  DMS-1312856, from ERC and FP7-People under grant 605096,
from the US-AFOSR under grant FA9550-12-10332, and from the Binational
(US-Israel) Science Foundation through grant 2010239.
P.G.K.'s work at Los Alamos is supported in part
by the U.S. Department of Energy. }
\date{\today}



\begin{abstract}
In the present work, we consider the mass in mass (or mass with mass)
system of granular chains, namely a granular chain involving additionally
an internal resonator. For these chains, we rigorously establish that
under suitable ``anti-resonance'' conditions connecting the mass of
the resonator and the speed of the wave, bell-shaped traveling wave
solutions continue to exist in the system, in a way reminiscent of the
results proven for the standard granular chain of elastic Hertzian
contacts. We also numerically touch upon settings where
the conditions do not hold, illustrating, in line also with recent
experimental work, that non-monotonic waves bearing non-vanishing
tails may exist in the latter case.
\end{abstract}

\maketitle

\section{Introduction}

A topic that has been progressively gaining more attention over the last
two decades within the nonlinear dynamical systems of the famous
Fermi-Pasta-Ulam type~\cite{FPU}, concerns the study of the so-called
granular crystals~\cite{Nesterenko2001,Sen2008,pgk_review,Theocharis_rev}.
The latter consist of chains of elastically interacting
(through the so-called Hertzian contacts) beads that are not only
very experimentally accessible, but also extensively tunable and
controllable, as regards their materials, geometry, heterogeneity, etc.
Another reason for the considerable appeal of such simple lattice
structure, which can be tailored to be one-, two- or even three-dimensional,
is the wealth of nonlinear excitations that have arisen in such settings,
and which include robust traveling waves, bright and dark breather
structures, and shock waves among others, as summarized in the above
reviews. Finally, yet another element promoting the interest in such
setups concerns their potential relevance to a broad range
of applications such as actuating devices \cite{dev08}, acoustic lenses~\cite{Spadoni}, mechanical diodes~\cite{china1,china2,Nature11},
logic gates~\cite{Feng_transistor2014} and sound
scramblers~\cite{dar05,Nesterenko2005}.

Arguably, the most significant and well-studied excitation, the principal
workhorse around which many theoretical analyses and experimental results
have been centered in such granular systems is the traveling wave.
This structure was originally established in the pioneering
work of Nesterenko in~\cite{Nesterenko1983} (see also
the review of~\cite{Nesterenko2001}). While it was originally
believed to be genuinely compact, asymptotic analysis illustrated
its doubly exponential nature (in the granular system without
precompression; in the presence of precompression the decay
of the wave becomes exponential)~\cite{chatter,pikovsky}.
>From a mathematically rigorous perspective, the work of~\cite{fw}
enabled as a special case example application the proof of existence
of such a traveling wave in the work of~\cite{mackay}, while later
the framework of~\cite{pego} (earlier also discussed in the
physical/computational
literature by~\cite{mertens}) enabled the rigorous identification
of the bell-shaped nature of the wave without~\cite{SteKev}
or with~\cite{SteKev2} precompression.

More recently, variants of the standard granular system in which
internal oscillators or resonators are present in each of the lattice
nodes (i.e., for each of the beads) have been proposed theoretically
and some of them have also been realized experimentally. Arguably,
one of the earliest examples of this type (which, however, has not
yet been experimentally implemented, to the best of our
knowledge) is the so-called cradle
system~\cite{ref28,ref29}, where a local, linear oscillator was
added to each bead, enabling the observation of some intriguing
dynamical features. More recently, another variant of the granular
chain, namely the locally resonant granular crystals or so-called
mass-in-mass or mass-with-mass systems have been proposed and
experimentally realized, respectively in~\cite{luca} and~\cite{gatz}.
The former setting involves an internal resonator within the chain
(e.g. a second bead embedded in the principal one), while in the
latter case, the resonator is external to the bead (see e.g. the prototype
also put forth in~\cite{annavain}).

The recent experiment of~\cite{yang} is expected to provide a significant
boost to studies along this direction, as it offered a different
so-called woodpile configuration
constructed out of orthogonally-stacked  rods
(with every second rod aligned, in this
alternating 0-90 degree configuration) for which it was
demonstrated that the MiM/MwM description of a granular chain with
an internal resonator can be tuned to be the
relevant one~\cite{yang2}. In fact, such a system allows,
depending on the parametric regime and the vibrational modes
of the rods, to create settings where controllably two- or
more internal resonator modes are relevant in the mathematical
model description. Moreover, this experiment reported the observation
of traveling waves for this system with persistent tails.

The focus of the present study is to provide a rigorous analysis
of the traveling waves in such a MiM/MwM system bearing an internal
resonator. In particular, the surprising finding that our analysis
will rigorously establish concerns the fact that under special
(i.e., non-generic or isolated)
so-called anti-resonance conditions, the traveling waves of this
chain will {\it still} bear a bell-shaped structure, monotonically
decaying on either side. That is to say, for such isolated combinations
of the resonator mass/wave speed, the tails of the wave will be absent.
On the contrary, the generic setting will, in fact, involve such a
resonant excitation through the passing the wave of the internal resonator
and as such, it will lead to the formation of the weakly nonlocal
solitary waves (nanoptera) reported in~\cite{yang}. Our presentation
of these results will be structured as follows. In section 2,
we will provide the general setup of the problem and an integral
(single component) reformulation thereof. In section 3, the main
result will be presented and proved. Finally, in section 4, some
supporting numerical computations will be provided in order to corroborate
the relevant results.

\section{Analysis of the problem: the model and its integral reformulation}

Following the earlier MiM/MwM studies~\cite{luca,gatz,annavain,yang,yang2},
we consider the prototypical mathematical model of the form:
\begin{equation}
\label{10}
\left|\begin{array}{l}
\p_{tt} X_i=[(X_{i-1}-X_i)^p_+- (X_{i}-X_{i+1})^p_+]+\tilde{k}(x_i-X_i)\\
\nu \p_{tt} x_i=-\tilde{k}(x_i-X_i)
\end{array}\right.
\end{equation}
 where $p>1$.
Here $X_i$ represent the displacements of the granular beads and
$x_i$ those of the internal (or external) resonators. $\tilde{k}$
corresponds to the (normalized) elastic constant of the bead-resonator
coupling, while $\nu$ is the (again, normalized to the mass of
the beads) mass of the resonator particles.
In the strain variables $Y_i:=X_{i-1}-X_i$, $y_i=x_{i-1}-x_i$, the equations take the form
 \begin{equation}
\label{20}
\left|\begin{array}{l}
\p_{tt} Y_i= (Y_{i+1}^p)_+ -2 (Y_i)^p_+ + (Y_{i-1})^p_+ + \tilde{k}(y_i-Y_i)\\
\nu \p_{tt} y_i=-\tilde{k}(y_i-Y_i)
\end{array}\right.
\end{equation}

 We first transform the problem to a setting which allows us to use the methods of calculus of variations. This is similar to the approach that we took in
the earlier works \cite{SteKev,SteKev2} to treat the case of monomer chains
(in the absence of resonators).

\subsection{Reduction to a single equation for $\Phi$}
We are now looking for traveling wave solutions of \eqref{20}  in the form
$Y_i(t)=\Phi(i-ct)$ and $y_i(t)=\Psi(i-c t)$, where we assume that $\Phi$ will be a positive function and we set hereafter for simplicity $\tilde{k}=1$.
 Plugging this ansatz in the system \eqref{20} yields
the following system of advance-delay differential equations
 \begin{equation}
\label{30}
\left|\begin{array}{l}
c^2 \Phi''= \De_{discr.}(\Phi^p) +  (\Psi-\Phi)\\
\nu c^2 \Psi''=- (\Psi-\Phi)
\end{array}\right.
\end{equation}
where, we have introduced the discrete Laplacian
$
\De_{disc} f(x)= f(x+1)-2 f(x)+f(x-1).
$

In order to restate the problem in its equivalent Fourier variable form,  we introduce the Fourier transform and its inverse as
\begin{eqnarray*}
 \hat{f}(\xi)= \int_{-\infty}^\infty f(x) e^{- 2\pi  i x \xi} dx;  \ \
 f(x)=    \int_{-\infty}^\infty \hat{f}(\xi) e^{ 2\pi i x \xi} d\xi
\end{eqnarray*}
 The second derivative operator $\p_x^2$ has the following  representation
$$
\widehat{\p_x^2 f}(\xi)=- 4\pi^2 \xi^2 \hat{f}(\xi).
$$
whereas
$$
\widehat{\De_{disc} f}(\xi)=-4\sin^2(\pi \xi) \hat{f}(\xi).
$$
Thus, taking Fourier transform in the second equation of \eqref{30} allows us to express
\begin{equation}
\label{40}
\widehat{\Psi}(\xi)= \f{1}{1-4\pi^2\nu c^2 \xi^2} \widehat{\Phi}(\xi).
\end{equation}
Plugging this last formula in the first equation of \eqref{30} then yields
the following equation for $\widehat{\Phi}$
$$
 -4\pi^2 c^2 \xi^2 \widehat{\Phi}(\xi) = -4\sin^2(\pi \xi)\widehat{\Phi^p}(\xi) + \f{4\pi^2\nu   c^2 \xi^2}{1-4\pi^2 \nu c^2 \xi^2}\widehat{\Phi}(\xi).
$$
 Solving for $\widehat{\Phi}$, we obtain
 \begin{equation}
 \label{50}
 \widehat{\Phi}(\xi)=   \f{1-4\pi^2 c^2 \nu \xi^2}{1+\nu-4\pi^2 c^2 \nu \xi^2}  \f{\sin^2(\pi \xi)}{ c^2 \pi^2\xi^2} \widehat{\Phi^p}(\xi)
 \end{equation}
 \subsection{An integral equation for $\Phi$}
With the  assignment $A=A_\nu:=\sqrt{1+\f{1}{\nu}}$, we can represent
 \begin{eqnarray*}
 \f{1-4\pi^2 c^2 \nu \xi^2}{1+\nu-4\pi^2 c^2 \nu \xi^2} =  1 -
 \f{\nu}{1+\nu-4\pi^2 c^2 \nu \xi^2} = 1+ \f{1}{2A}
 \left(\f{1}{2\pi c \xi - A}- \f{1}{2\pi c \xi + A}\right),
 \end{eqnarray*}
which allows us to further rewrite \eqref{50} as follows
\begin{equation}
\label{60}
 \widehat{\Phi}(\xi)= \f{\sin^2(\pi \xi)}{ c^2 \pi^2\xi^2} \widehat{\Phi^p}(\xi)+
 \f{1}{2A}
 \left(\f{1}{2\pi c \xi - A}- \f{1}{2\pi c \xi + A}\right) \f{\sin^2(\pi \xi)}{ c^2 \pi^2\xi^2} \widehat{\Phi^p}(\xi).
 \end{equation}
  Recall that (see \cite{pego} and also \cite{SteKev}, \cite{SteKev2})
 taking the ``tent'' function, $\La (x)=(1-|x|)_+$  or
$$
\La(x)=\left\{\begin{array}{l l}
1-|x| & |x|\leq 1, \\
0 & |x|>1.
\end{array}\right.
$$
we have
 $\hat{\La}(\xi)= \f{\sin^2(\pi \xi)}{\pi^2 \xi^2}$.

 Next, we compute the inverse Fourier transform of $\f{1}{2A}
 \left(\f{1}{2\pi c \xi - A}- \f{1}{2\pi c \xi + A}\right)$. We have
 \begin{eqnarray*}
 \int \f{1}{2\pi c \xi - A} e^{2\pi i \xi x} d\xi = \f{e^{i \f{A}{c} x}}{2\pi c} i \int \f{\sin(z x)}{z} dz= \f{e^{i \f{A}{c} x}}{2 c} i {\rm sgn}(x),
 \end{eqnarray*}
 whence
 $$
 \cf^{-1} \left[\f{1}{2A}
 \left(\f{1}{2\pi c \xi - A}- \f{1}{2\pi c \xi + A}\right)\right]=
 -\f{1}{2 A c} \sin\left(\f{A}{c} x\right) {\rm sgn}(x).
 $$
 Thus, taking inverse Fourier transform in \eqref{60}, we obtain the
analogous to~\cite{pego} (see also~\cite{mertens})
representation for the solution of the form:
 \begin{equation}
 \label{70}
 c^2 \Phi=   \La *  \Phi^p - \f{1}{2 A c } \sin\left(\f{A}{c} x\right)
{\rm sgn}(x)*\La* \Phi^p.
 \end{equation}
The next task is to compute the kernel
 $$
- \f{1}{2 A c } \sin\left(\f{A}{c} x\right) {\rm sgn}(x)*\La.
 $$
 Clearly, this is an even function, being the convolution of two even functions. Thus, we only need to compute it for $x>0$ and then we can take an even extension across zero. A direct computation shows
 $$
 - \f{1}{2 A c } \sin\left(\f{A}{c} x\right) {\rm sgn}(x)*\La=  \left\{
  \begin{array}{ll}
  \f{x-1}{A^2}+ \f{c}{A^3}\left( \sin(\f{A}{c}) \cos(\f{A}{c} x) - \sin(\f{A}{c}x)\right)&  x\in (0,1). \\
  \f{c \sin\left(\f{A x}{c}\right)(\cos\left(\f{A}{c} \right) - 1)}{A^3} & x\geq 1
  \end{array}
  \right.
 $$
 At this point, we introduce a few more notations in order to rewrite \eqref{70} in more compact form. Namely, let
 $$
\mu:=\frac{A}{c}>0,\ \   G(x)=G_{\mu, A}(x):= \left\{
  \begin{array}{ll}
   \f{ \sin(\mu)\cos(\mu x)-\sin(\mu x) \sgn(x)}{\mu}&  x\in (-1,1). \\
   \f{ \sin\left(\mu x\right)(\cos(\mu)-1)}{\mu} \sgn(x) &  |x|\geq 1.
  \end{array}
  \right.
 $$
{\it Note that the function $G$ is compactly supported only if $\mu=2n \pi$.}

 We can now rewrite \eqref{70} in the form
 \begin{equation}
 \label{110}
 A^2 c^2 \Phi= (A^2-1) \La * \Phi^p  +  G*\Phi^p.
 \end{equation}
For positive solutions of \eqref{110}, we can take the transformation $Z=\Phi^p$, which leads us to
 \begin{equation}
 \label{120}
 A^2 c^2 Z^{1/p}= (A^2-1) \La * Z  + G*Z=((A^2-1)\La+G)*Z.
 \end{equation}
 Denote
 $$
 K(x):=(A^2-1)\La+G(x).
 $$
 Note that $K=K(A, c;x)$ and the problem \eqref{130} now reads
 \begin{equation}
 \label{140}
 A^2 c^2 Z^{1/p}= K_{A,c}*Z.
 \end{equation}
 \subsection{The anti-resonance condition}
Hereafter, we restrict ourselves to the case of compactly supported kernel $K$ in the convolution problem \eqref{120}. That is, we require
\begin{equation}
\label{130}
\mu=2 n \pi, n\in {\mathbf N}
\end{equation}
 in order to achieve that the function $G$ (and hence $K$) is supported in $(-1,1)$. We refer to \eqref{130}
  as the {\bf anti-resonance condition} for the parameters. Notice that
this is a condition that physically
connects the mass of the resonator (or effectively the ratio of its mass
to that of the principal bead) to the speed of the wave.
We can then obtain the following
  \begin{lemma}
  \label{le:15}
  Let $n\in {\mathbf N}, \mu=2\pi n$. There exists an irrational number $A_0=1.10328...$, so that for $A\geq A_0$, the kernel
  $$
  K_{A,c}(x)=(A^2-1) \La(x)-\f{\sin(\mu x)}{ \mu} sgn(x)\chi_{(-1,1)}(x),
  $$
  is positive in $(0,1)$. For $A>\sqrt{2}$, the function $K$ is   decreasing in $(0,1)$.
  \end{lemma}
  \begin{proof}
  Since $K$ is obviously an even function, it suffices to consider the case $0\leq x\leq 1$.
  Taking the derivative of $K$, we see that
$$
K'(x)=-\cos(\mu x)-(A^2-1) \leq 1-(A^2-1)\leq 2-A^2\leq 0
$$
so long as $A\geq \sqrt{2}$. Thus, $K$ is decreasing in $(0,1)$ for $A\geq \sqrt{2}$.

  Next, we study  the positivity of $K$. Again, it suffices to consider $x\in (0,1)$.
We have
  $$
  K(x)=(A^2-1) (1-x) - \f{\sin(\mu x)}{\mu}=(1-x)\left[A^2-1-\f{\sin(\mu x)}{\mu (1-x)}\right].
  $$
  Since $\mu=2n\pi$, we can     rewrite the last expression as follows
  $$
  K(x)=(1-x) \left[A^2-1+\f{\sin(2\pi n(1-x))}{2\pi n (1-x)}\right].
  $$
  Now, since
$$
-a_0=\inf_{z\in (0,2\pi)}  \f{\sin z}{z}=\inf_{z\in (0,\infty)}  \f{\sin z}{z}=-0.217234....,
$$
we conclude that $K(x)\geq 0$, provided $A\geq A_0=\sqrt{1+a_0}=1.10328...$

  \end{proof}

We can now proceed to establish our main result.

\section{Main result: Bell-shaped Traveling Waves persist Under
Suitable Anti-Resonance Conditions}
We have the following main result.
\begin{theorem}
\label{theo:main}
Let $\frac{A}{c}=2 n \pi$ for some integer $n=1, 2, \ldots$.

Then, for $A\geq \sqrt{2}$, the equation \eqref{140} (and hence \eqref{50}) has bell-shaped solution $\Phi$. That is, there is $\Phi:{\mathbf R}\to {\mathbf R}_+$ a positive, even, $C^\infty$
smooth function, so that $\Phi$ is decaying in $(0, \infty)$. In addition, $\Phi$ is doubly exponentially decaying, just as the travelling waves in the classical monomer case.

For $\sqrt{2}>A\geq A_0=1.10328...$, the problem \eqref{140} (and hence \eqref{50}) has a positive solution $\Phi$, which  we cannot guarantee to be
bell-shaped~\footnote{The latter feature will, however, be illustrated
in our case example numerical computations of the next section.}.
\end{theorem}
We perform the proof in several steps, under the different assumptions on $A$. Following the idea in \cite{SteKev, SteKev2}, we let $q=1+\f{1}{p}$ and we set up the following maximization problem
\begin{equation}
\label{200}
\left\{
\begin{array}{l}
\dpr{K*z}{z}\to \max \\
\int_{\rone} |z(x)|^{q} dx=1
\end{array}
\right.
\end{equation}
Note that $q\in (1,2)$.
The plan of actions is as follows:  we first  show that \eqref{200} has a solution with the desired properties. After that,  we derive its Euler-Lagrange equation. This is of course closely related to \eqref{140}, except for the Lagrange multipliers, which need to be adjusted.
\subsection{Existence for the maximizers}
We will show that the following lemma holds.
\begin{lemma}
\label{le:14}
Let $A\geq A_0$. Then, the maximization problem \eqref{200} has a   solution. Moreover, this solution is positive. If in addition  $A\geq \sqrt{2}$, then   it is also bell-shaped.
 \end{lemma}
First, we show that the quantity $\dpr{K*z}{z}$ is bounded, if $z$ satisfies the constraint. Indeed, by H\"older's and Young's inequality\footnote{Note that since $q=1+\f{1}{p}<2$, we have that $\f{q}{2q-2}>1$.}
$$
|\dpr{K*z}{z}|\leq \|K*z\|_{L^{q'}} \|z\|_{L^q}\leq \|K\|_{L^{\f{q}{2q-2}}} \|z\|_{L^q}^2,
$$
But $K$ is a bounded function with  support in $(-1,1)$, hence $K\in L^{\f{q}{2q-2}}$.
Next, we show that an eventual solution to  \eqref{200} is necessarily positive. Indeed, for any function $z$, we have that the function $w:=|z|$ satisfies the constraint and moreover
$$
\dpr{K*z}{z}=\int \int K(x-y)z(x) z(y) dx dy\leq \int \int K(x-y)w(x) w(y) dx dy=\dpr{K*w}{w}.
$$
Thus, the supremum in \eqref{200} may be taken over all $z\geq 0$. In fact, one can see that if $z$ is not non-negative, then $w=|z|$ provides a bigger value and hence $z$ may not be a solution to \eqref{200}.
Denote
$$
J^{\max}=\sup_{\|z\|_{L^q}=1} \dpr{K*z}{z}.
$$
Clearly, $J^{\max}>0$.
Pick a maximizing sequence, say $z_n: z_n\geq 0$. That is $\|z_n\|_{L^q}=1$ and
$$
\dpr{K*z_n}{z_n}\to J^{\max}.
$$
We will show that (an appropriate translate of) $z_n$ converges strongly (in $L^q$) to a solution $z_0$. To that end,  we apply the concentration compactness
method of Lions, \cite{Lions}. The outcome is that (after we pick a subsequence, which we call again $z_n$), one of three scenarios occur:
\begin{itemize}
\item (tightness) There exists $y_k\in \rone$, so that for each $\epsilon>0$, there is $R=R(\epsilon)$, so that
$$
\int_{y_k-R}^{y_k+R} z_k^q(x) dx >1-\epsilon
$$
\item (vanishing) For every $R>0$, there is
$$
\lim_{k\to \infty} \sup_y \int_{y-R}^{y+R} z_k^q(x) dx = 0.
$$
\item (dichotomy) There exists $\al\in (0,1)$, so that for any $\epsilon>0$, there are $R=R(\epsilon)$ and $R_k\to \infty$, $y_k\in \rone$, so that for all large enough $k$,
$$
\left|\int_{y_k-R}^{y_k+R} z_k^q(x) dx - \al\right|<\epsilon^q, \  \left|\int_{|y-y_k|>R_k} z_k^q(x) dx -(1- \al)\right|<
\epsilon^q
$$
\end{itemize}
We will proceed to show that vanishing and dichotomy may not occur, which will leave us with tightness. In the tightness scenario, we will easily show that a translate of $z_k$ will converge strongly to $z_0$. \\
{\bf Vanishing does not occur:}\\
Assume that it does. Let $\epsilon>0$, $R=10$ and $k_0=k_0(\epsilon)$ is so large that
\begin{equation}
\label{210}
\sup_y  \int_{y-10}^{y+10} z_k^q(x) dx <\epsilon^q.
\end{equation}
for all $k>k_0$. We have for each $y_0\in \rone$,
$$
\int_{y_0-5}^{y_0+5} K*z_k(y) z_k(y) dy\leq \|z_k\|_{L^q(y_0-5,y_0+5)} \| K*z_k\|_{L^{q'}(y_0-5,y_0+5)}.
$$
Since $supp K\subset (-1,1)$, and the integration is over $(y_0-5, y_0+5)$, it follows that
$$
\chi_{(y_0-5, y_0+5)} K*z_k(y) = \chi_{(y_0-5, y_0+5)} K*[z_k\chi_{(y_0-6, y_0+6)}].
$$
We have by Young's inequality that
$$
\| K*z_k\|_{L^{q'}(y_0-5,y_0+5)}\leq \|K\|_{L^{\f{q}{2q-2}}} \|z_k\|_{L^q(y_0-6,y_0+6)}.
$$
It follows that
$$
\int_{y_0-5}^{y_0+5} K*z_k(y) z_k(y) dy\leq C \|z_k\|_{L^q(y_0-10,y_0+10)}^2\leq  C \epsilon^{2-q}
\|z_k\|_{L^q(y_0-10,y_0+10)}^q,
$$
where in the last inequality, we have used \eqref{210}. Summing the last inequality over $y_0={0, \pm 1, \ldots}$ yields
$$
\dpr{K*z_k}{z_k}\leq \sum_{y_0\in {\mathcal Z}} \int_{y_0-5}^{y_0+5} K*z_k(y) z_k(y) dy\leq  C_1 \epsilon^{2-q}  \|z_k\|_{L^q}= C_1 \epsilon^{2-q},
$$
for all large enough $k$ and for all $\epsilon>0$.
On the other hand $\dpr{K*z_k}{z_k}\to J^{\max}$, which will be a contradiction, if we have had selected  $\epsilon$ small enough.  Note that again, we have used $q<2$. Thus, vanishing may not occur.
\\
{\bf Dichotomy  does not occur:}\\
Assume that it does. Let $\epsilon>0$ and denote
$$
z_k^1(y):= z_k(y) \chi_{(y_k-R, y_k+R)}(y); \ \  z_k^2(y)= z_k(y) \chi_{|y-y_k|>R_k}(y),
$$
so that $\|z_k-z_k^1-z_k^2\|_{L^q}=O(\epsilon)$.  From support considerations and the Young's estimates that we have provided earlier, it is clear that for all large enough  $k$
$$
\dpr{K*z_k}{z_k}=\dpr{K*z^1_k}{z^1_k}+ \dpr{K*z^2_k}{z^2_k}+O(\epsilon).
$$
Now, from the definition of $J^{\max}$
$$
\dpr{K*z^1_k}{z^1_k}= \|z^1_k\|_{q}^2\dpr{K*\f{z^1_k}{\|z^1_k\|_{q} }}{\f{z^1_k}{\|z^1_k\|_{q}}}\leq J^{\max} \|z^1_k\|_{q}^2=J^{\max}(\al^2+O(\epsilon^q)).
$$
Similarly,
$$
\dpr{K*z^2_k}{z^2_k}\leq J^{\max}((1-\al)^2+O(\epsilon^q)).
$$
Putting everything together yields
$$
\dpr{K*z_k}{z_k}\leq J^{\max}(\al^2+(1-\al)^2+O(\epsilon^q))+O(\epsilon).
$$
This is again produces a contradiction (with  a judiciously
small choice of $\epsilon$), since \\ $\dpr{K*z_k}{z_k}\to J^{\max}$, $\al^2+(1-\al)^2<1$.   Thus, dichotomy does not occur either.

Thus, tightness is the only alternative. Now, consider $\tilde{z}_n(y):=z_n(y-y_n)$. Note that $\tilde{z}_n$ satisfies the constraint and also $\dpr{K* \tilde{z}_n}{\tilde{z}_n}=\dpr{K*z_n}{z_n}\to J^{\max}$. We have that for all $\epsilon>0$, there exists $R=R(\epsilon)$, so that
\begin{equation}
\label{250}
\int_{-R}^R \tilde{z}_n^q(y) dy>1-\epsilon.
\end{equation}
 for all $n$.
Now, by weak compactness in $L^q$, it follows that $\tilde{z}_n$ has a weakly convergent subsequence (denoted again by $\tilde{z}_n$ for convenience), with a limit say $z_0$.

We now claim that the sequence $\{K*\tilde{z}_n\}_n$ is strongly pre-compact in $L^{q'}$.  Indeed, by the Young's inequality, we have for all $n$
$$
\|K*\tilde{z}_n\|_{W^{1,q'}}\leq \|K*\tilde{z}_n\|_{L^{q'}}+ \|K'*\tilde{z}_n\|_{L^{q'}}\leq
C \|K\|_{W^{1,\f{q}{2q-2}}} \|z_n\|_{L^q}= C \|K\|_{W^{1,\f{q}{2q-2}}},
$$
where by  inspection $K\in  W^{1,\f{q}{2q-2}}$. Also
$$
\int_{|y|>R+1} |K*\tilde{z}_n(y)|^{q'}dy\leq  \|K\|_{L^{\f{q}{2q-2}}}^{q'} \|\tilde{z}_n\|_{L^q(|y|>R)}^{q'}
\leq C \epsilon^{q'/q}.
$$
By the compactness criteria in $L^{q'}$ spaces (i.e. the Riesz-Tamarkin condition), it follows that $\{K*\tilde{z}_n\}_n$ is pre-compact. This, together with the pointwise limit $K*\tilde{z}_n\to K*z_0$, which follows from the weak convergence of $\tilde{z}_k$ implies that for some subsequence
$$
\|K*\tilde{z}_{n_l}-K*z_0\|_{L^{q'}}\to 0.
$$
We now have by H\"older's
\begin{eqnarray*}
|\dpr{K*\tilde{z}_{n_l}}{\tilde{z}_{n_l}}-\dpr{K*z_0}{z_0}| &\leq & |\dpr{K*\tilde{z}_{n_l}-K*z_0}{\tilde{z}_{n_l}}|+
|\dpr{z_0}{K*\tilde{z}_{n_l}-K*z_0}|\\
&\leq & \|K*\tilde{z}_{n_l}-K*z_0\|_{L^{q'}} (\|z_0\|_{L^q}+\|\tilde{z}_{n_l}\|_{L^q})\to 0.
\end{eqnarray*}
Thus,
$$
\dpr{K*z_0}{z_0}=J^{\max}.
$$
On the other hand, by the lower semicontinuity of the norm with respect to the weak convergence
$\|z_0\|_{L^q}\leq 1$. But then $\|z_0\|_{L^q}=1$, since otherwise
$$
J^{\max}=\dpr{K*z_0}{z_0}=\|z_0\|_{L^q}^2  \dpr{K*\f{z_0}{\|z_0\|_{L^q}}}{\f{z_0}{\|z_0\|_{L^q}}}\leq J^{\max} \|z_0\|_{L^q}^2<J^{\max},
$$
a contradiction. Thus, we have shown that the limit $z_0$ is indeed a solution to the maximization problem \eqref{200}.
\subsection{Bell-shapedness of the solution in the case $A\geq \sqrt{2}$}
In this case,  we have that the kernel $K$ is bell-shaped.  We  show now  that the solution $z_0$ is in addition bell-shaped. In order  to explain the setup, we need a few definitions.

For a function $f:\rone\to \rone$, one defines the distribution function
$$
d_f(\al):=meas\{x: |f(x)|>\al\}.
$$
Clearly, the function $d_f(\al)$ is non-increasing, whence one can define its ``inverse'' as follows
$$
f^*(t)=\inf\{s: d_f(s)\leq t\}.
$$
We call $f^*:\rone_+\to \rone_+$ {\it the non-increasing rearrangement} of $f$. Note that the two functions have the same distribution function, that is $d_f(\al)=d_{f^*}(\al)$, so in particular $\|f\|_{L^p}=\|f^*\|_{L^p}: 0< p\leq \infty$. Finally, define the even function
 $f^\#(t)=f^*(2|t|)$, which also satisfies $\|f\|_{L^p}=\|f^\#\|_{L^p}: 0< p\leq \infty$.
 Clearly, a function is bell-shaped if and only if $f^\#=f$.  In this setting, we have the Riesz rearrangement inequality\footnote{The original inequality appeared implicitly in the work  of Riesz.}, \cite{BLM},   , which states
$$
\int_{\rone} \int_{\rone} f(x-y) g(y) h(x) dx dy\leq  \int_{\rone} \int_{\rone} f^\#(x-y) g^\#(y) h^\#(x) dx dy.
$$
Using that fact and since $K$ is bell-shaped, it is clear that in the maximization problem \eqref{200},
we can restrict our attention to bell-shaped entries $z$. Indeed, taking an arbitrary function $z$, note that $z^\#$ would satisfy the constraint $\|z^\#\|_{L^q}=\|z\|_{L^q}=1$ and moreover, by the Riesz rearrangement inequality
$$
\dpr{K*z}{z}=\int_{\rtwo}   K(x-y) z(y) z(x) dy dx\leq  \int_{\rtwo}  K(x-y) z^\#(y) z^\#(x) dx dy=
\dpr{K*z^\#}{z^\#}.
$$
Based on the above formulation, one may derive the existence of $z_0$ the same way as before. Note however that one may completely circumvent the Lions concentration compactness arguments in the bell-shaped case. Indeed, assuming that $z$ is bell-shaped and satisfies the constraint $\|z\|_{L^q}=1$, we have for every $x>0$,
$$
|z(x)|^q x\leq \int_0^x |z(y)|^q dy\leq \int_0^\infty |z(y)|^q dy=1/2.
$$
Hence $|z(x)|\leq 2^{-q} x^{-1/q}$. Thus, any maximizing sequence $\{z_n\}$ of bell-shaped functions will produce a pre-compact in $L^{q'}$ sequence $K*z_n$. Indeed, we have, similarly to the arguments above that $K*z_n\in W^{1,q'}$. In addition, by support considerations for any $R>10$, we have for every $n$
$$
\int_{|x|>R} |K*z_n(x)|^{q'}dx\leq C \int_{|x|>R-1} |z(x)|^{q'} dx\leq C \int_{|x|>R-1} \f{1}{|x|^{q'/q}} dx\leq C R^{1-\f{q'}{q}}\to 0,
$$
as $R\to \infty$, since $q<2<q'$.  Thus, the elementary pointwise bounds obtained from bell-shapedness help us  verify immediately the Riesz-Tamarkin criteria and hence the pre-compactness of $K*z_n$ in $L^{q'}$, without having to resort to the Lions' theory. After that, we finish   by standard arguments as before.
 \subsection{The Euler-Lagrange equation for \eqref{200}}
The next step is to derive the Euler-Lagrange equation that the solution $z_0$ of \eqref{200}  satisfies. We proceed as follows. Take $\epsilon\in \rone$ and a test function $h$. Since $z_0$ is a maximizer for \eqref{200}, it must be that
$$
\dpr{K*\f{z_0+\epsilon h}{\|z_0+\epsilon h\|_{L^q}}}{\f{z_0+\epsilon h}{\|z_0+\epsilon h\|_{L^q}}}\leq J^{\max}
$$
for all $\epsilon,h$. We can rewrite this as follows
\begin{equation}
\label{d:10}
\dpr{K*(z_0+\epsilon h)}{z_0+\epsilon h}\leq J^{\max} \|z_0+\epsilon h\|_{L^q}^2
\end{equation}
Taking Taylor expansions in $\epsilon$, we find
\begin{eqnarray*}
\dpr{K*(z_0+\epsilon h)}{z_0+\epsilon h} &=& \dpr{K*z_0}{z_0}+2\epsilon\dpr{K*z_0}{h}+O(\epsilon^2)=\\
&=&  J^{\max}+2\epsilon\dpr{K*z_0}{h}+O(\epsilon^2), \\
\|z_0+\epsilon h\|_{L^q}^2 &=& (1+\epsilon q \int z_0^{q-1}(x) h(x) dx+O(\epsilon^2))^2=1+2\epsilon q \dpr{z_0^{q-1}}{h}+O(\epsilon^2).
\end{eqnarray*}
Plugging in this in \eqref{d:10}, we obtain
$$
2\dpr{K*z_0}{h}=2 J^{\max} q \dpr{z_0^{q-1}}{h}+O(\epsilon),
$$
which is valid for all $\epsilon,h$. Taking limit as $\epsilon\to 0$, we obtain
$$
\dpr{K*z_0-q z_0^{q-1}}{h}=0.
$$
Since this is satisfied for all test functions $h$ and $q-1=\f{1}{p}$, we obtain
$
K*z_0=q z_0^{\f{1}{p}}.
$
If we now take $z_0=\mu Z$,
$$
K*Z=q\mu^{\f{1}{p}-1} Z^{\f{1}{p}}.
$$
Clearly, with an appropriate choice of $\mu$, we can arrange so that $Z$ indeed satisfies \eqref{140}.

We now proceed to show that $Z$ and $\Phi=Z^p$ are both $C^\infty$. First, let us see that from the equation \eqref{140}, the function $Z$ never vanishes. Indeed, assume the opposite, $Z(x_0)=0$. Then,
$$
0=A^2c^2 Z^{1/p}(x_0)=\int_{-1}^1 K(y) Z(x-y)dy.
$$
Since $K(y)>0$ in $(-1,1)$, it would follow that $Z|_{(x_0-1,x_0+1)}=0$. Iterating this argument yields that $Z(x)=0$ for all $x$, a contradiction. Now, we have from \eqref{140} that $\Phi$ is differentiable  and moreover
$$
A^2 c^2 \Phi'=K'*\Phi^p.
$$
It follows that $\Phi'$ is continuous. By induction, we can show
$$
A^2 c^2 \Phi^{(l)}=K'*(\Phi^p)^{(l-1)},
$$
and hence $\Phi^{(l)}$ is continuous, since $\Phi^{(j)},j<l$ are all continuous and $\Phi\neq 0$. For integer values of $p$, this is obvious.
Even for non-integer values of $p$, $\Phi^p$ is as smooth as $\Phi$, since $\Phi(x)>0$.

Finally, let us establish the double exponential rate of decay for the bell-shaped solutions of \eqref{50}. So, assuming $A\geq \sqrt{2}$, we have that $\Phi$  is bell-shaped and $K$ is decaying for $x>0$. We have
$$
A^2 c^2 \Phi(x)=\int_{x-1}^{x+1} K(y) \Phi^p(x-y) dy.
$$
Denote $\tilde{\Phi}:= \left(\f{2\|K\|_{L^\infty}}{A^2 c^2}\right)^{\f{1}{p-1}} \Phi$, so that
$$
\tilde{\Phi}(x)=\f{1}{2\|K\|_{L^\infty}} \int_{x-1}^{x+1} K(y) \tilde{\Phi}^p(x-y) dy.
$$
Using the bell-shapedness of both $K$ and $\Phi$  we conclude
$$
\tilde{\Phi}(x)\leq \tilde{\Phi}^p(x-1).
$$
Iterating this inequality for  $n<x$,
$$
\tilde{\Phi}(x)\leq  \tilde{\Phi}^{p^n}(x-n).
$$
which yields the double exponential rate of decay for $\tilde{\Phi}$ and hence for $\Phi$.  This completes the proof of Theorem \ref{theo:main}.


In order to corroborate the above conclusions, as well as to examine
the cases which do not adhere to the conditions of the above theorem
(both cases of anti-resonance, but with $A<\sqrt{2}$, as well
as cases where the anti-resonance condition is not satisfied), we
turn to some case example numerical computations in the next section.

\section{Numerical results}

\subsection{Anti-resonance scenario ($\mu=2\pi n$)}

When $\mu=2\pi n$, $n\in\mathbb{N}$, we obtained the expression of kernel $K_{A,c}(x)$ (as shown in Lemma 1) in the convolution equation: $$A^2 c^2 \Phi=K_{A,c}*(\Phi)^p.$$
When a solution of $\Phi$ is obtained, one can similarly find $\Psi$ from:
 \begin{equation}
 \label{260}
 \widehat{\Psi}(\xi)=   \f{1}{1+\nu-4\pi^2 c^2 \nu \xi^2}  \f{\sin^2(\pi \xi)}{ c^2 \pi^2\xi^2} \widehat{\Phi^p}(\xi)
 \end{equation} and its corresponding integral equation
 \begin{equation}
 \label{270}
 {\Psi}= M_{A,c}*{\Phi^p}
 \end{equation} where
$$
M_{A,c}(x)=\frac{1}{c^2\nu A^2}[(1-|x|)+\frac{\sin(\mu x)}{\mu}\sgn(x)]\chi_{(-1,1)}(x).
$$
Using the above expressions,
the solution for both the bead and the resonator strains
can be not only theoretically analyzed but also numerically computed.

Due to the similarity between $M_{A,c}$ and $K_{A,c}$, we will mainly discuss the properties of $K_{A,c}$ and those of $M_{A,c}$ directly follow.
Since $K_{A,c}$ is an even function and has finite support $(-1,1)$, we only need to consider its properties on $(0,1)$.
Noting  that $K'_{A,c}(0^+)=K'_{A,c}(1^-)=-A^2$, $K'_{A,c}(0^-)=K'_{A,c}(-1^+)=A^2$,
we appreciated that $K_{A,c}$ is continuous everywhere but not continuously
differentiable at $x=-1,0,1$. We now separate the different cases of
interest.

\begin{enumerate}
\item The case analyzed fully in the previous section
(i.e., the ``best case scenario''):
when $A>\sqrt{2}$, $K'_{A,c}(x)<0$ for $x\in(0,1)$ so that $K_{A,c}$ is increasing on $[-1,0]$ and decreasing on $[0,1]$. This regime
also implies the positivity of $K_{A,c}$ since $K_{A,c}(-1)=K_{A,c}(1)=0$. Due to these good features of $K_{A,c}(x)$, and as a by-product of the proof
of the previous section, our numerical computations confirm
that the system bears bell-shaped solutions for $\Phi$ and $\Psi$,
as illustrated in the typical results shown in Fig.~\ref{fig1}.

\begin{figure}[!htbp]
\begin{tabular}{cc}
\includegraphics[width=7cm]{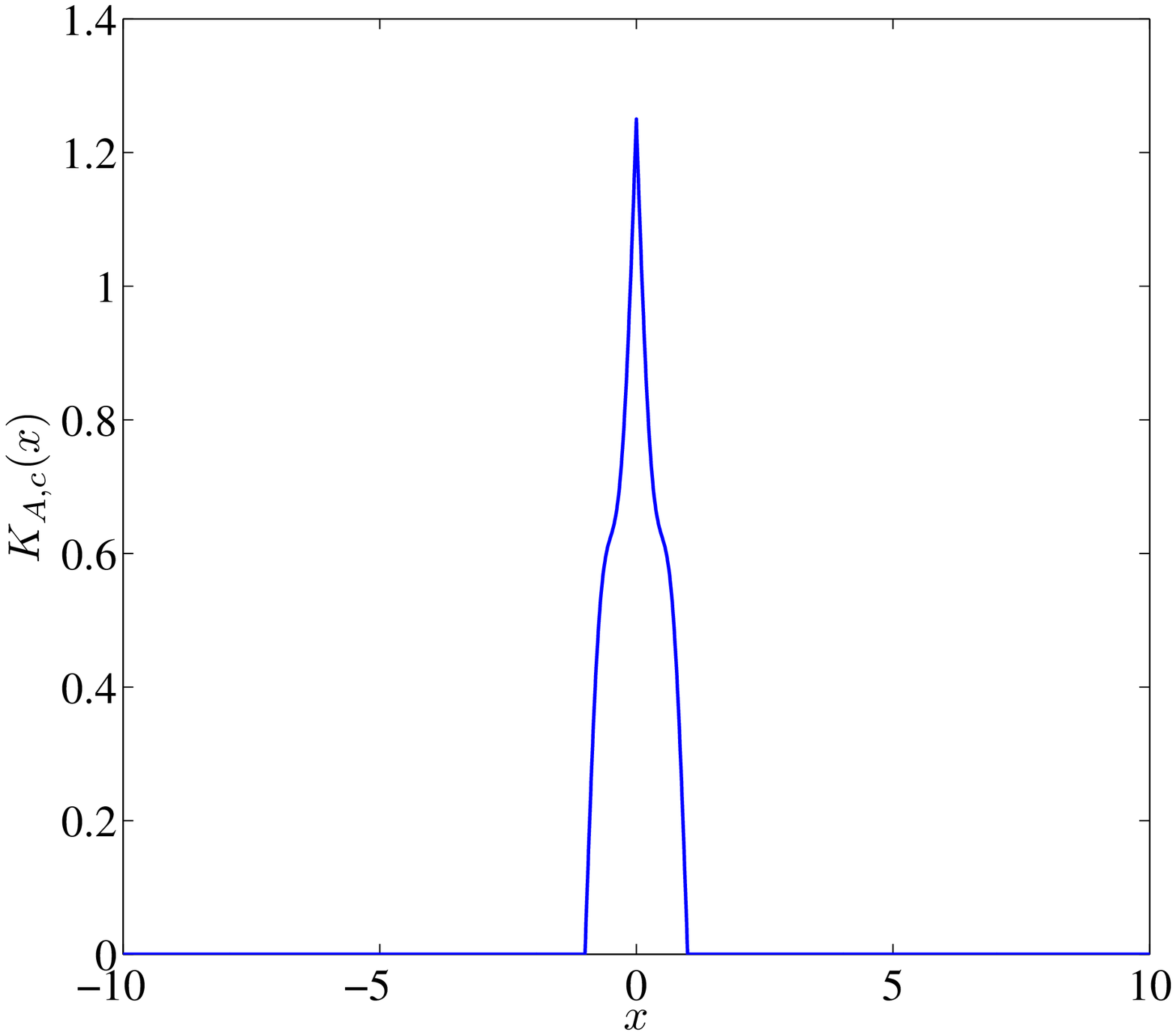}
\includegraphics[width=7cm]{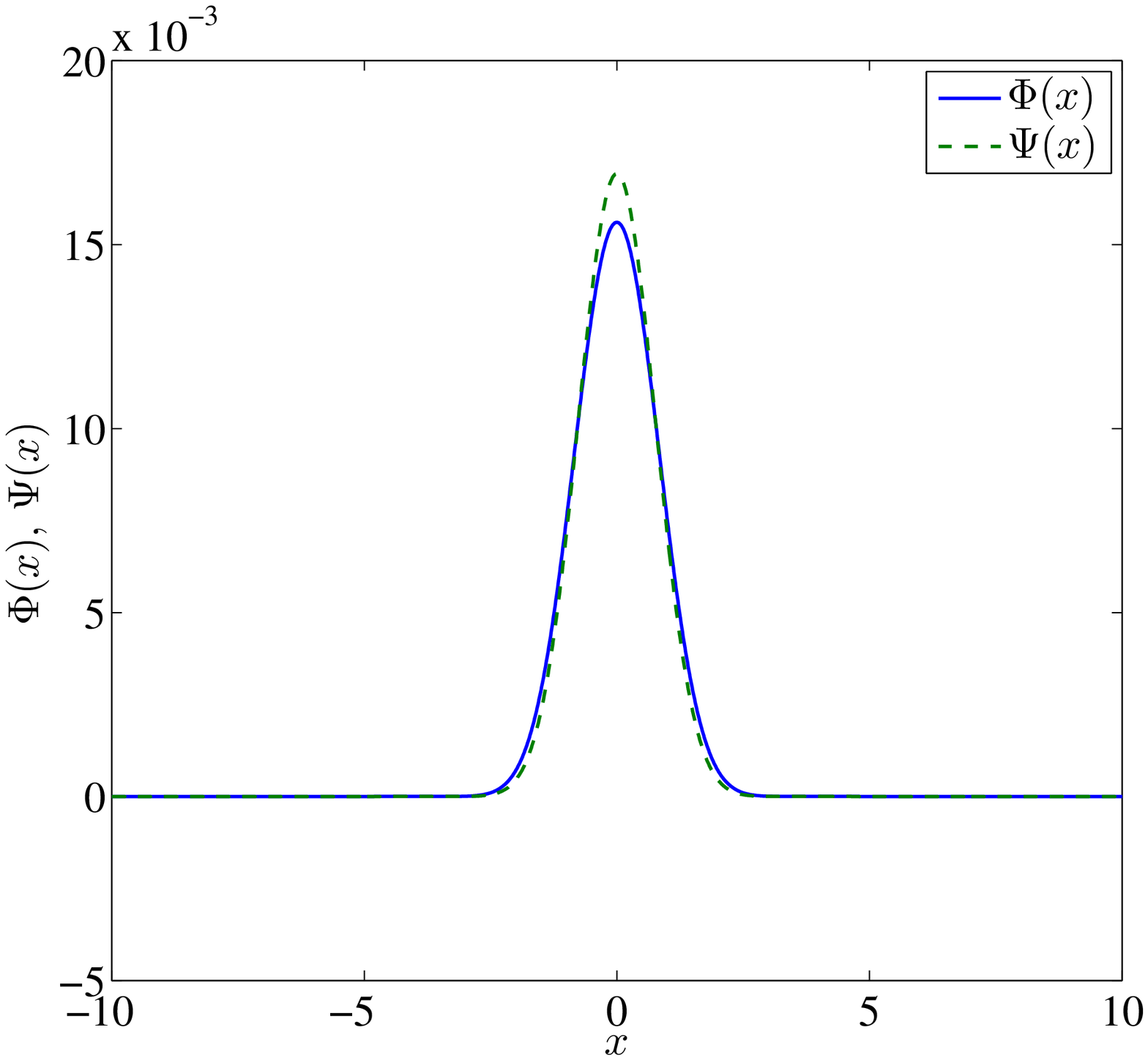}\\
\includegraphics[width=7cm]{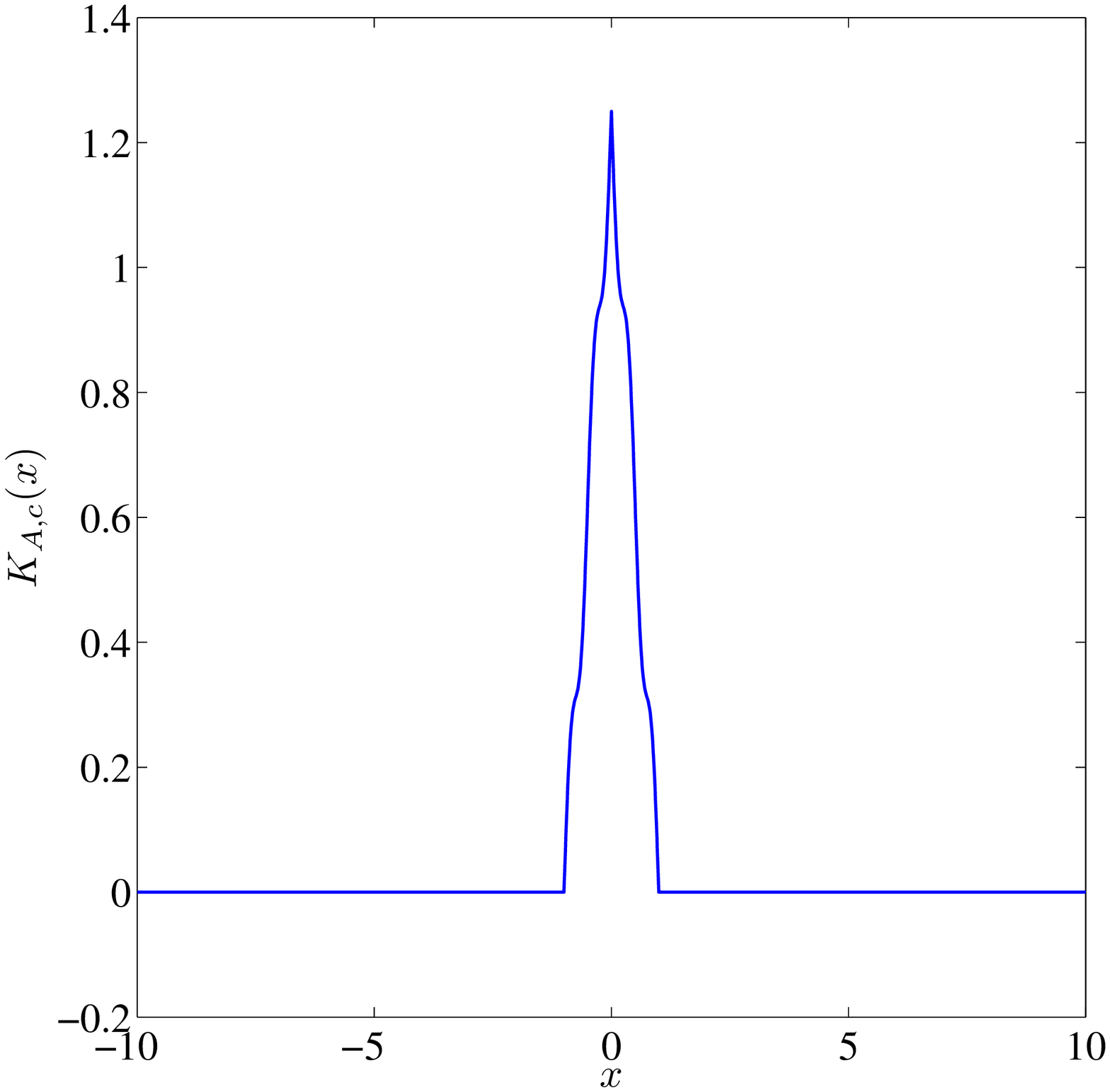}
\includegraphics[width=7cm]{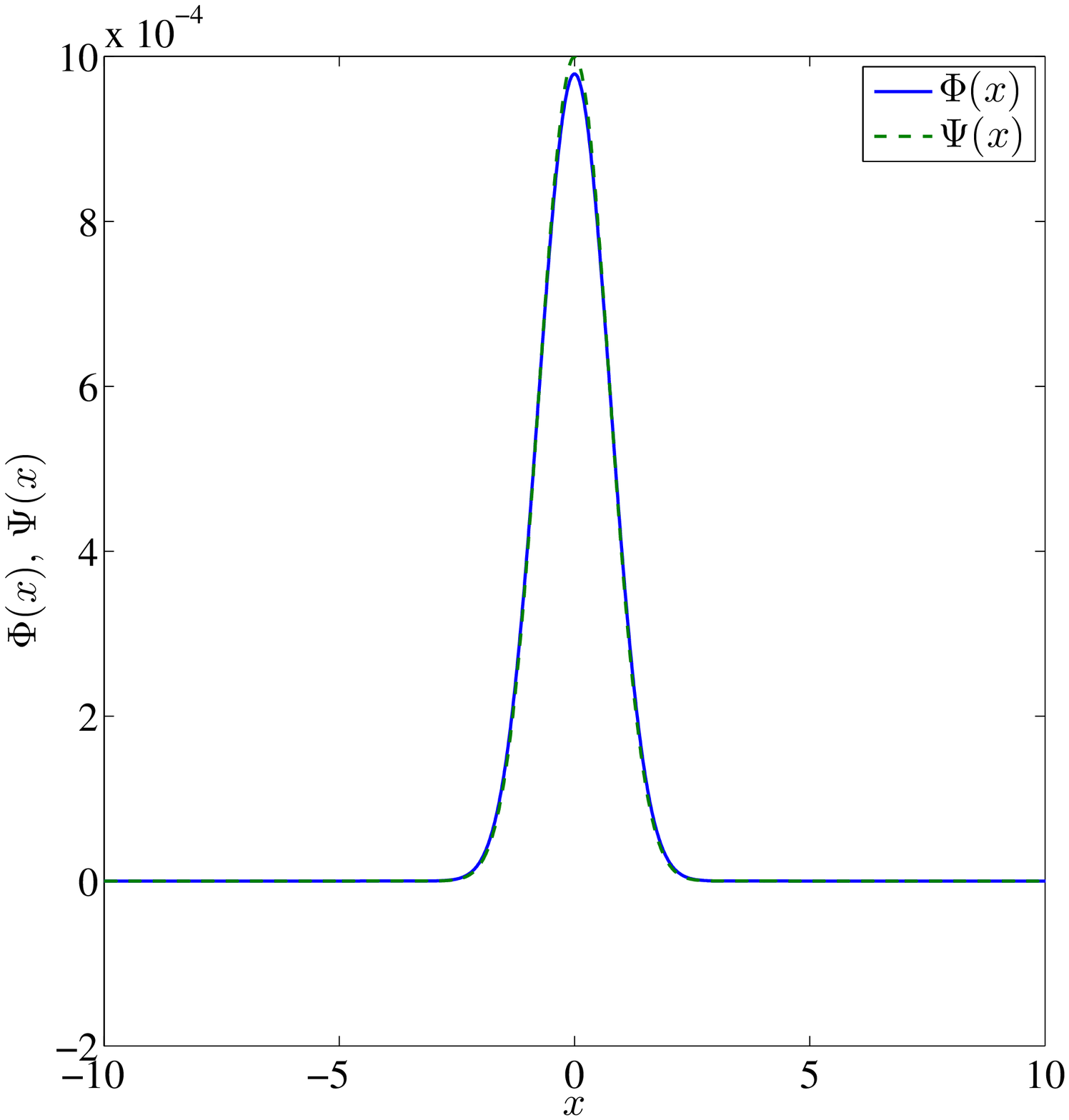}\\
\includegraphics[width=7cm]{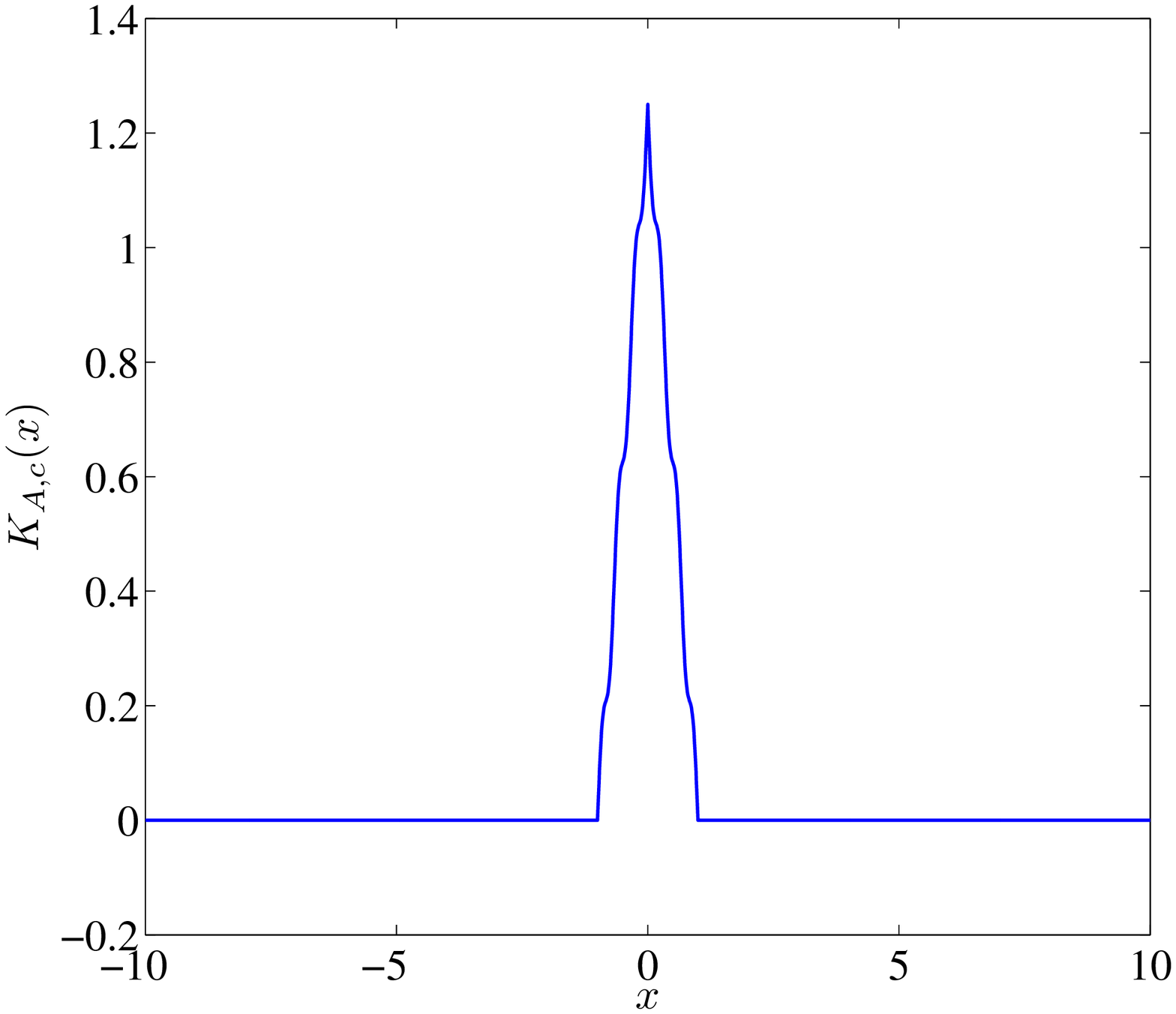}
\includegraphics[width=7cm]{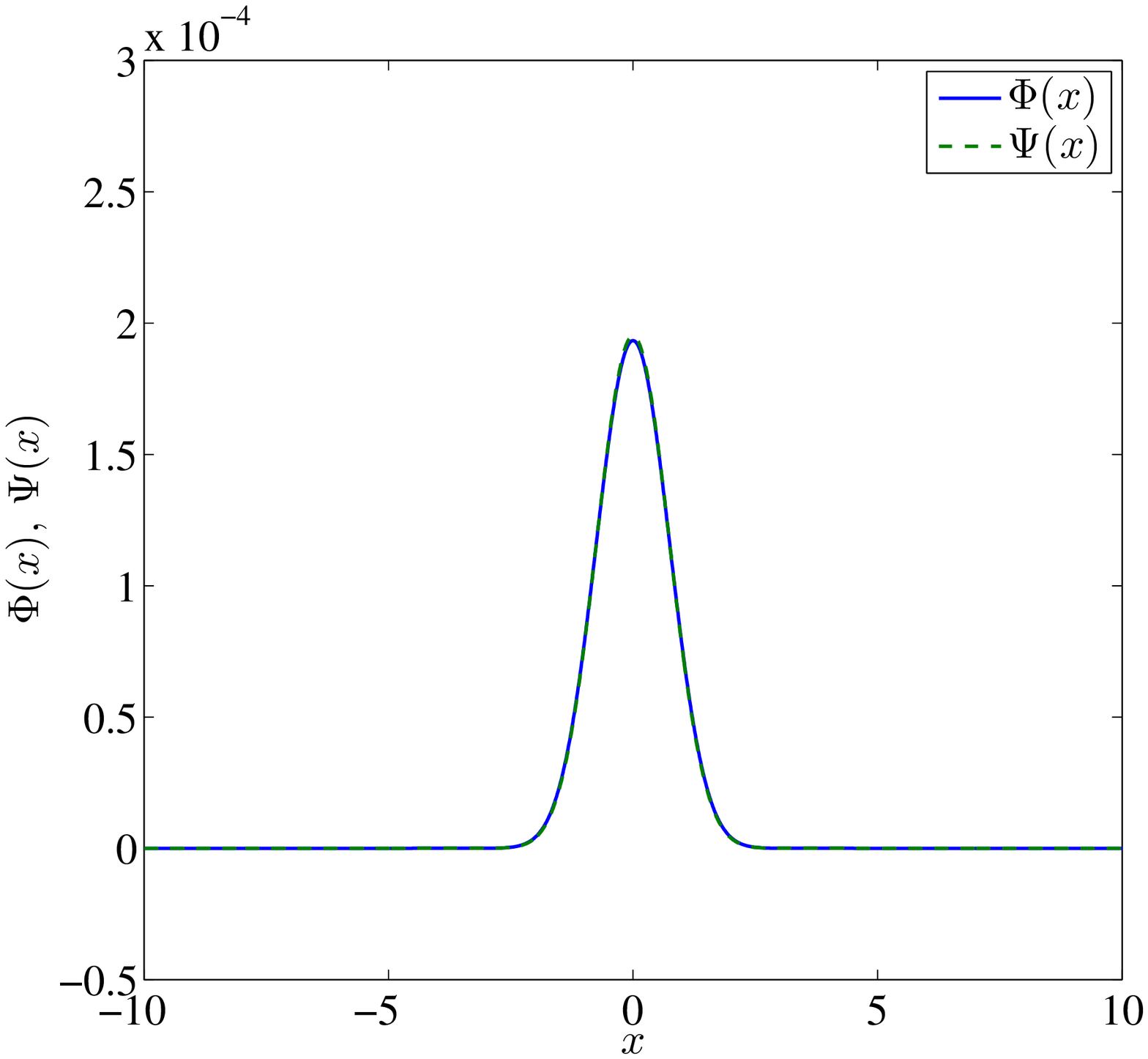}
\end{tabular}
\caption{The top panels show the kernel of the convolution equation $K_{A,c}(x)$ and corresponding solution $\Phi$ and $\Psi$ with $A=1.5$ and $\mu=2\pi$. In the middle panels and bottom panels, $\mu$ is set as $4\pi$ and $6\pi$,
respectively, corresponding to higher anti-resonances.}
\label{fig1}
\end{figure}

\item The ``intermediate'' case: when $1.10328\approx A_0<A<\sqrt{2}$, $K_{A,c}(x)\geq0$ still holds for all $x$ but the kernel is no longer
monotone on $(-1,0)$ or $(0,1)$. Although in this category
$K_{A,c}(x)$ has non-monotonic variations (with
the number of local minima being $n$), the corresponding solution of $\Phi$
still features a bell-shaped profile as it did in the previous case
(see Fig.~\ref{fig2} for such examples in the case of different
anti-resonances). The difference between the (non-monotonic)
kernel and the (monotonic) solution
is, apparently, caused by the smoothing effect of convolution.

\begin{figure}[!htbp]
\begin{tabular}{cc}
\includegraphics[width=7cm]{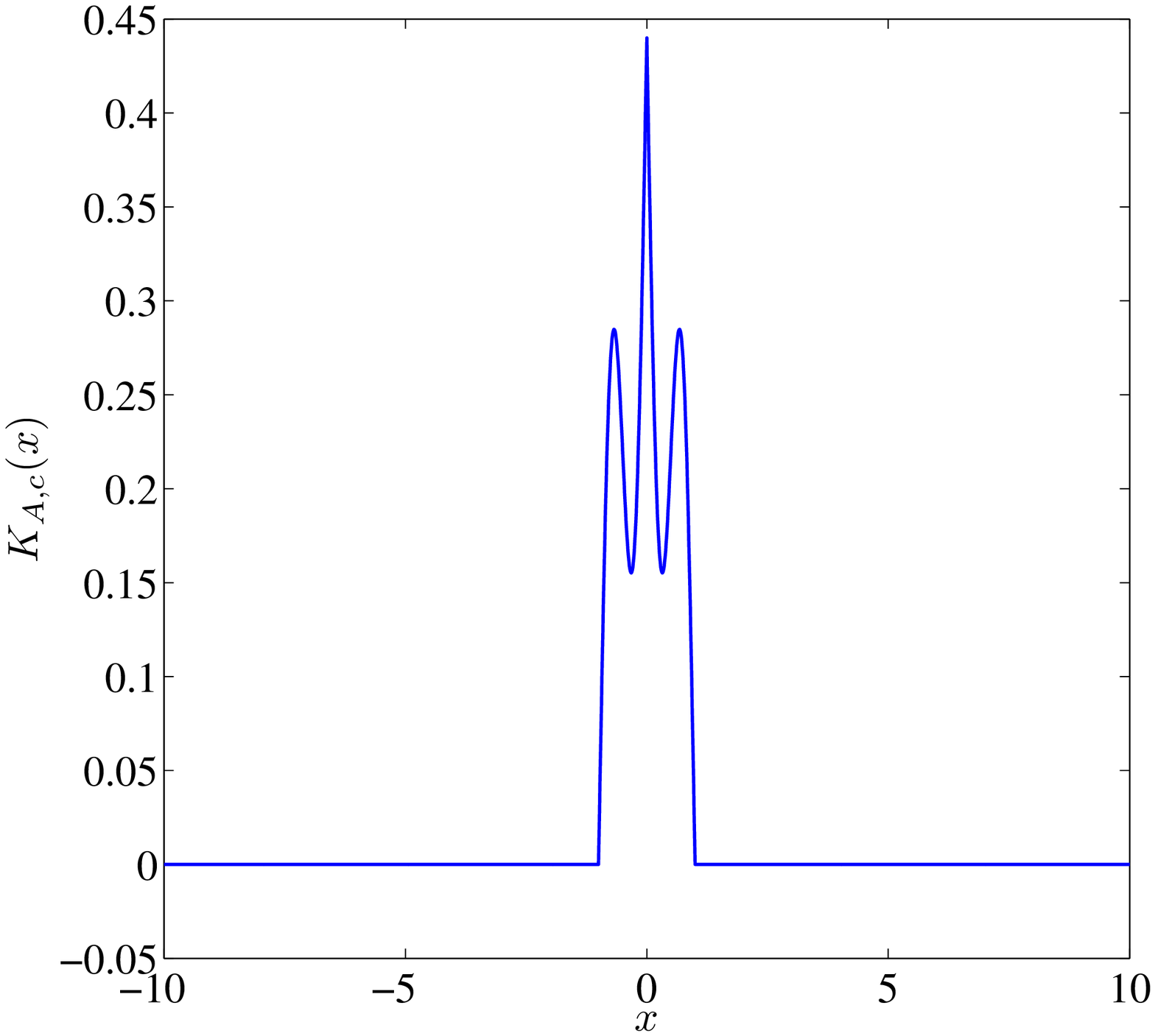}
\includegraphics[width=7cm]{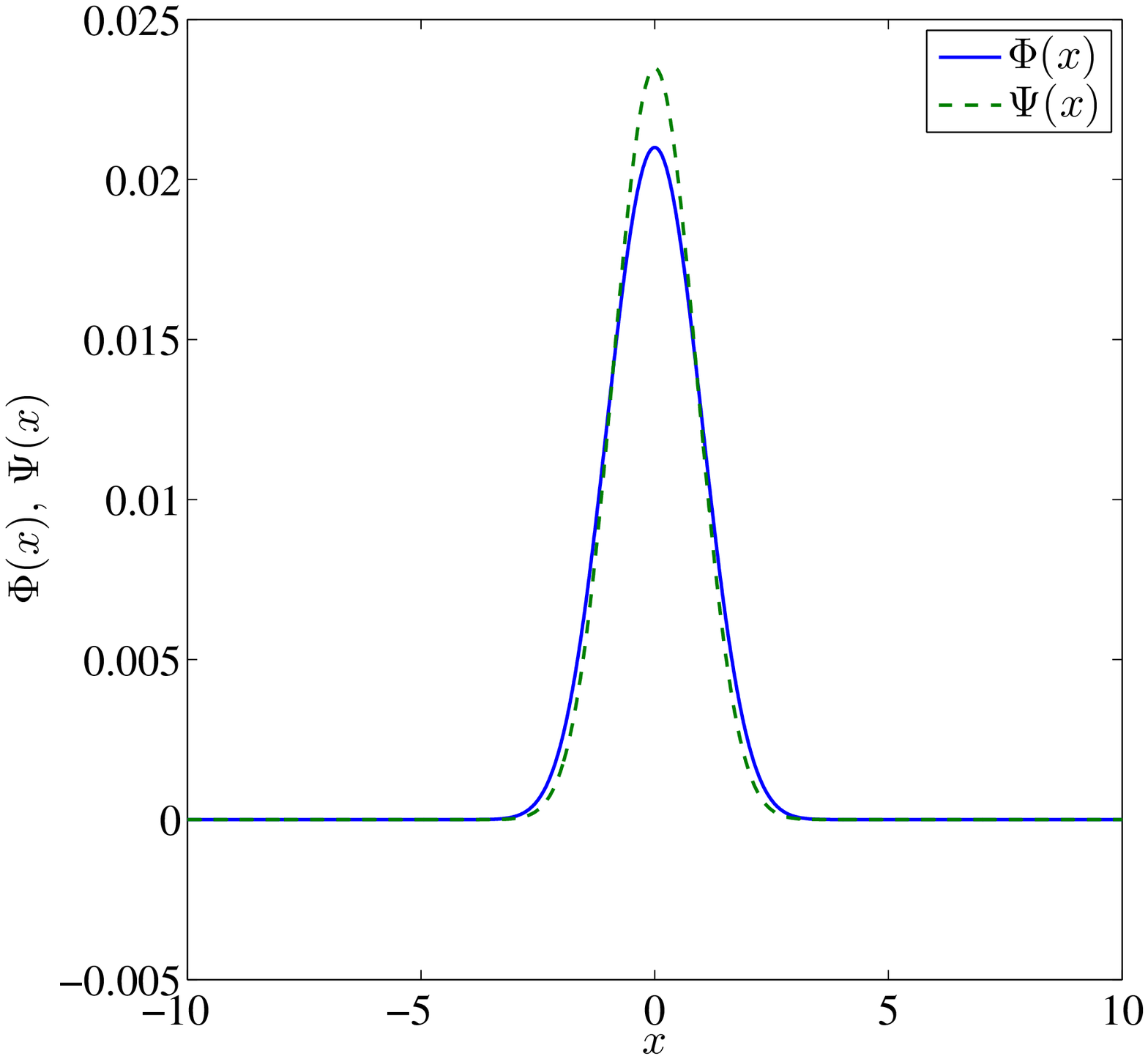}\\
\includegraphics[width=7cm]{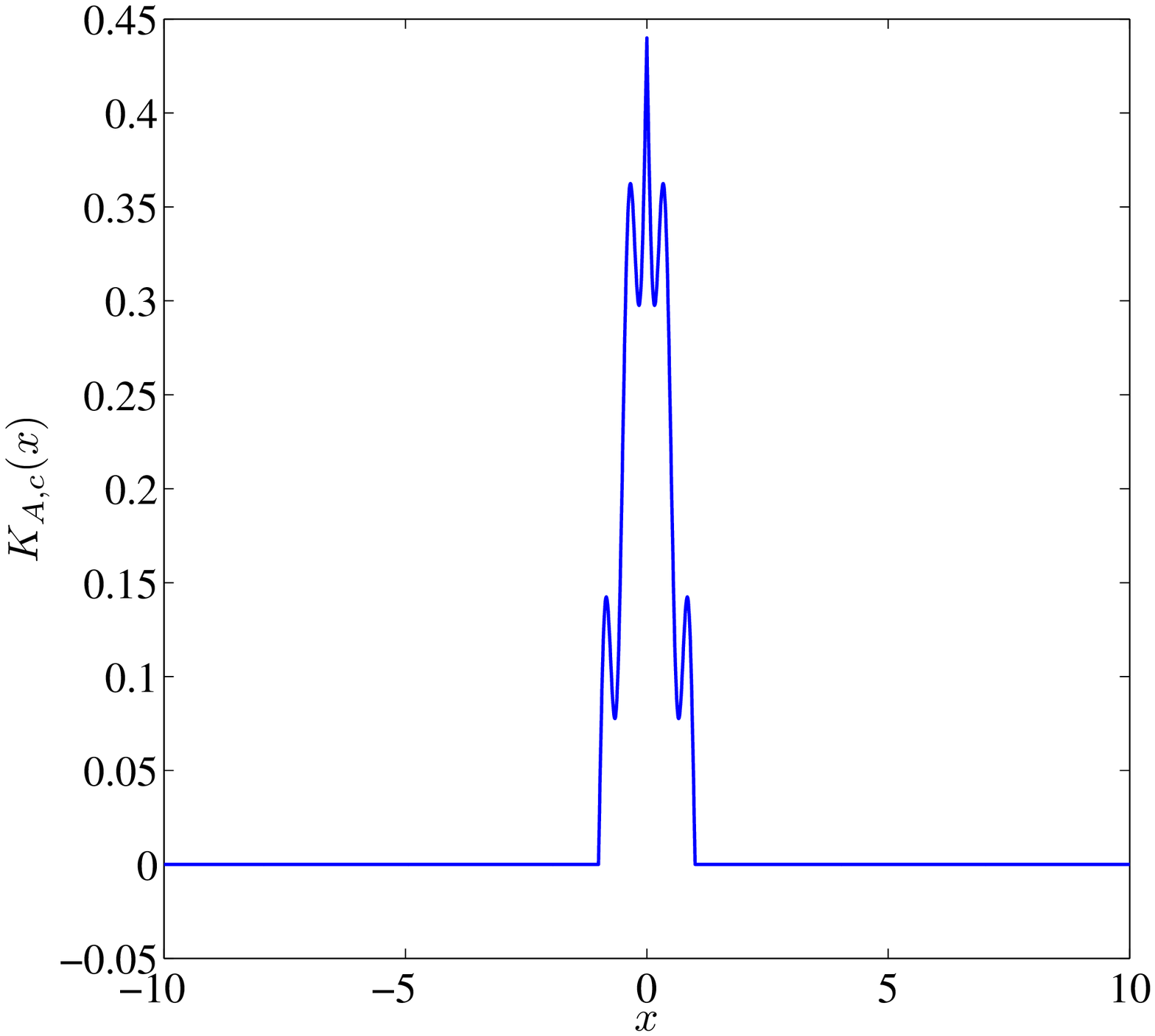}
\includegraphics[width=7cm]{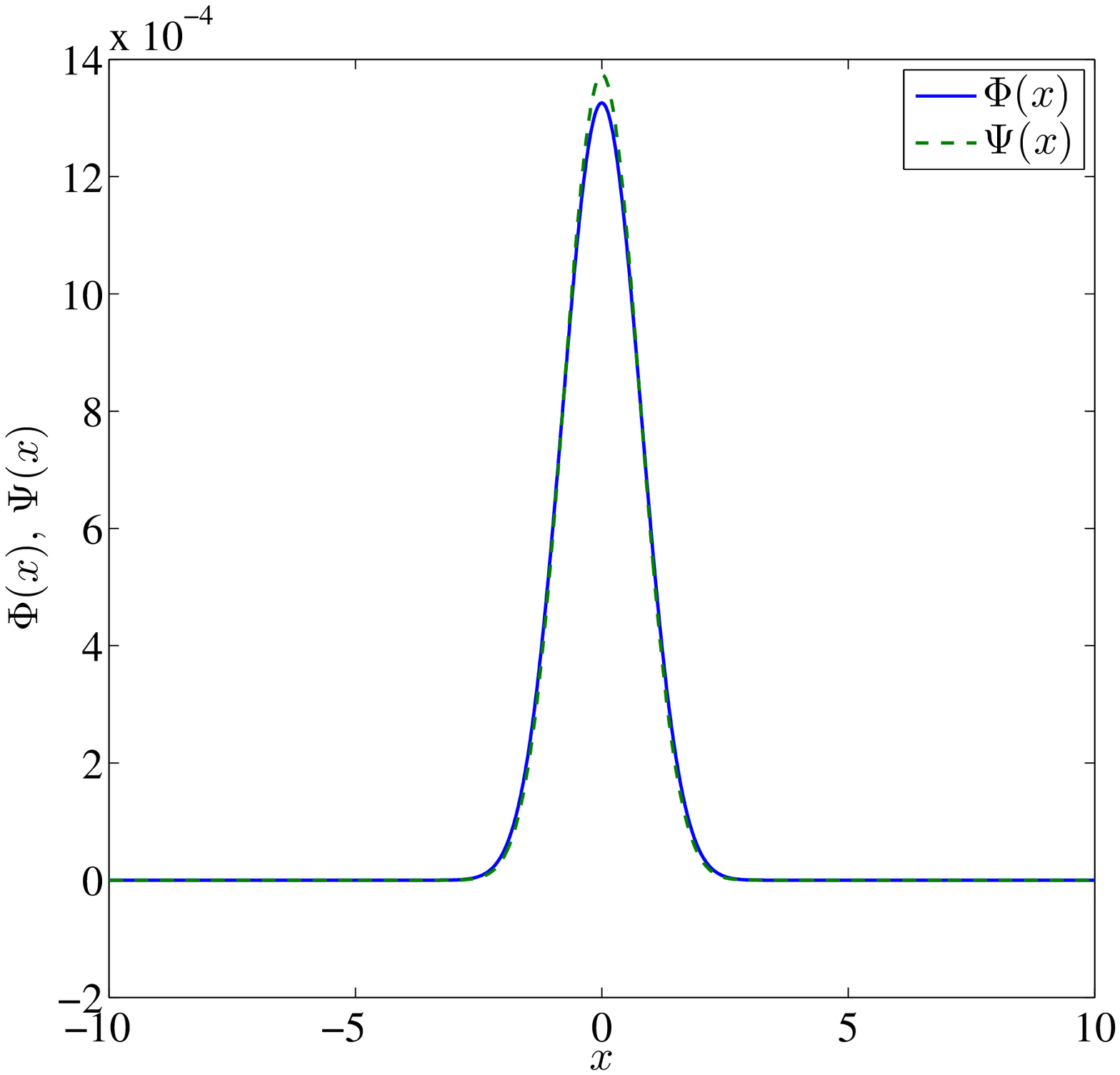}\\
\includegraphics[width=7cm]{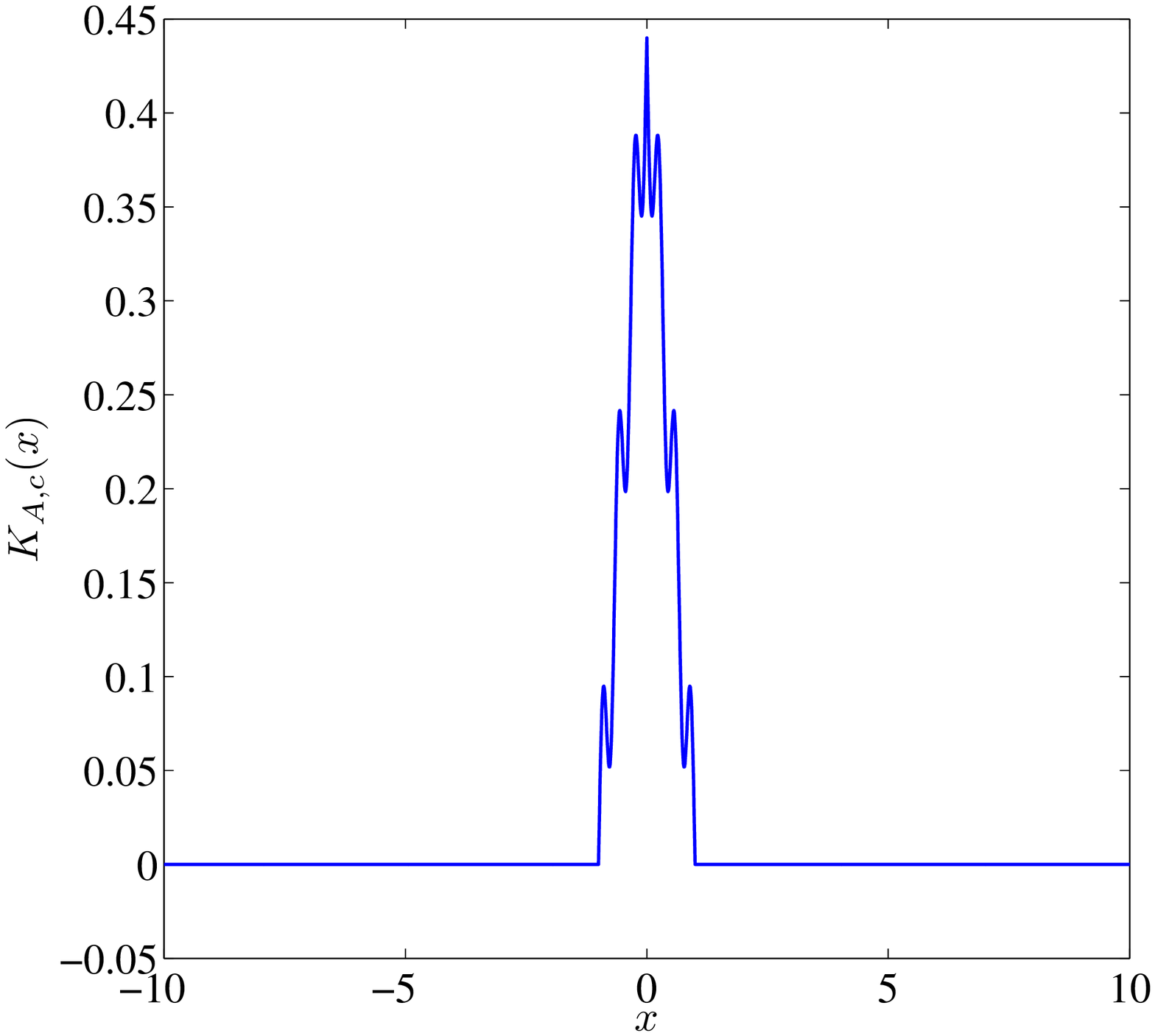}
\includegraphics[width=7cm]{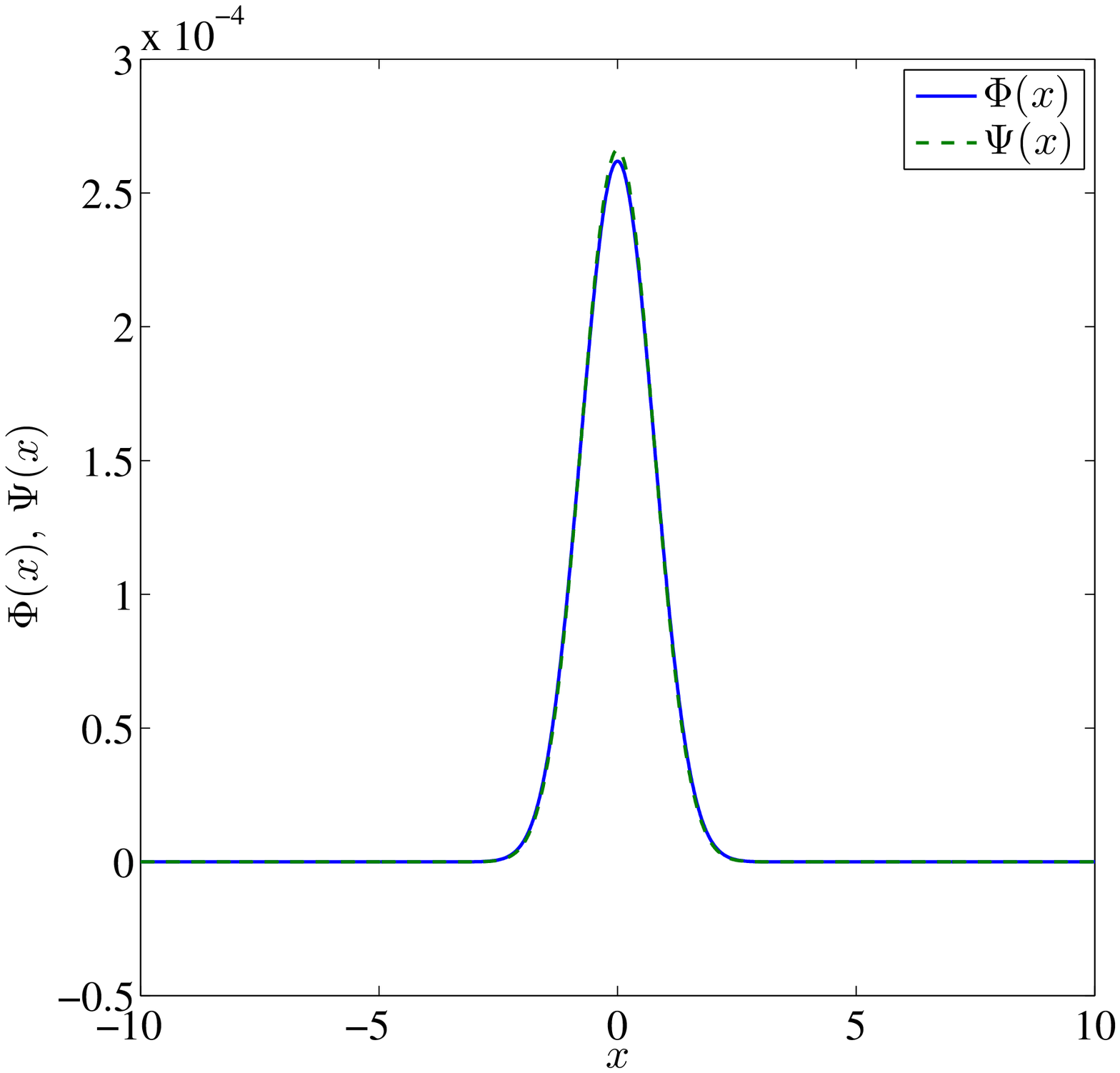}
\end{tabular}
\caption{This figure is structured similarly as Fig.~\ref{fig1} but $A$ has changed to $1.2$ instead of $1.5$ here. The top panels show the kernel of the convolution equation $K_{A,c}(x)$ and corresponding solution $\Phi$ and $\Psi$ with $A=1.2$ and $\mu=2\pi$. In the middle panels and bottom panels, $\mu$ is set as $4\pi$ and $6\pi$, respectively.}
\label{fig2}
\end{figure}

\item Turning now to cases for which our analytical arguments
can no longer be used (even in a qualitative way),
we first consider: $1<A_{1,n}<A<A_0$. Here,
$K_{A,c}(x)$ is neither necessarily positive, nor decreasing on $(0,1)$.
However, as Fig.~\ref{fig3} suggests, we notice from the numerical results
that a bell-shaped solution of $\Phi$ and $\Psi$ could still exist as long as $A$ is not too small, i.e. $A>A_{1,n}>1$. Here $A_{1,n}$ decreases over $n$
and it is the threshold above which our numerical method is able to
converge to a
solution, as shown in the right panel of Fig.~\ref{fig4}.

\begin{figure}[!htbp]
\begin{tabular}{cc}
\includegraphics[width=7cm]{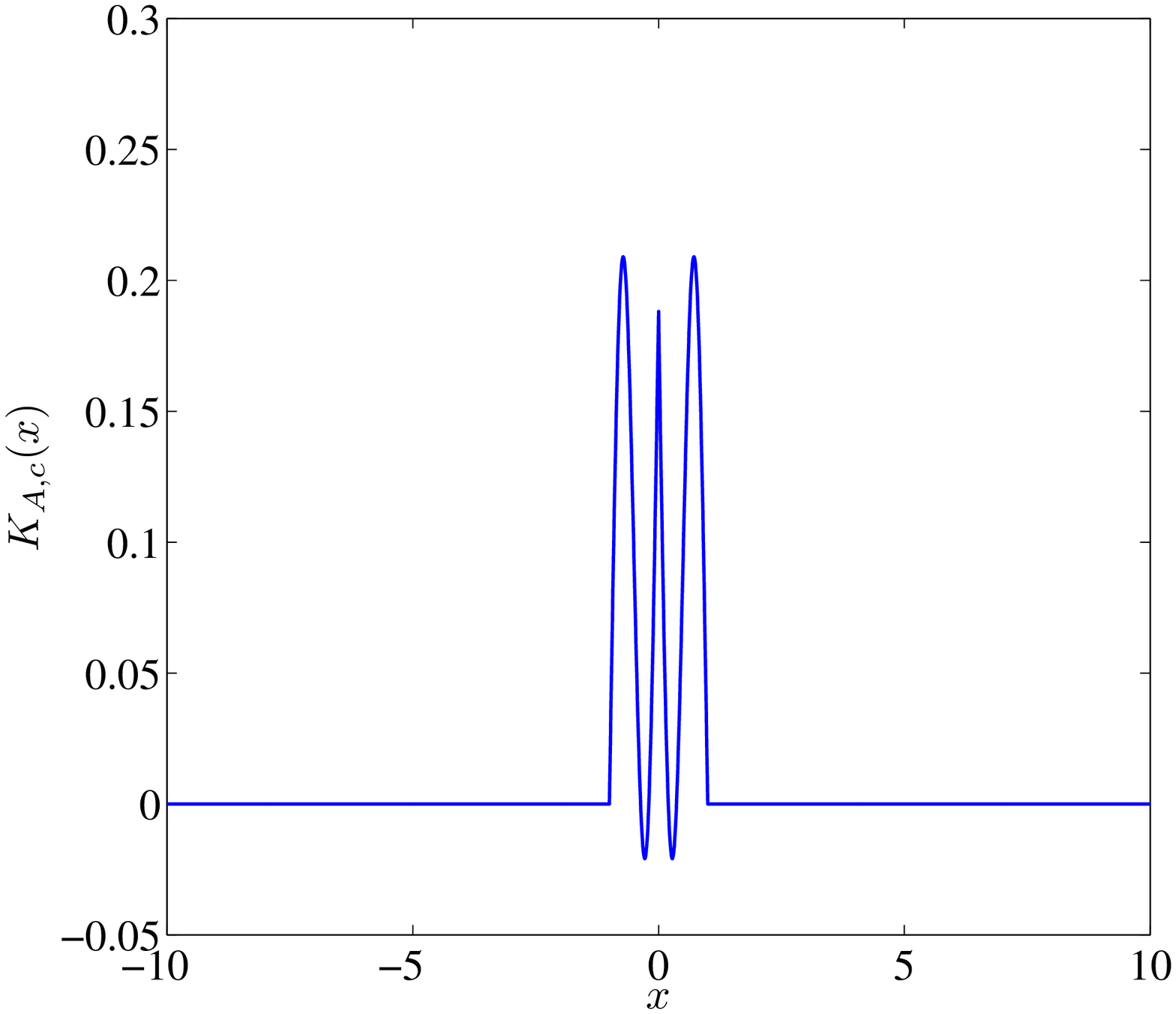}
\includegraphics[width=7cm]{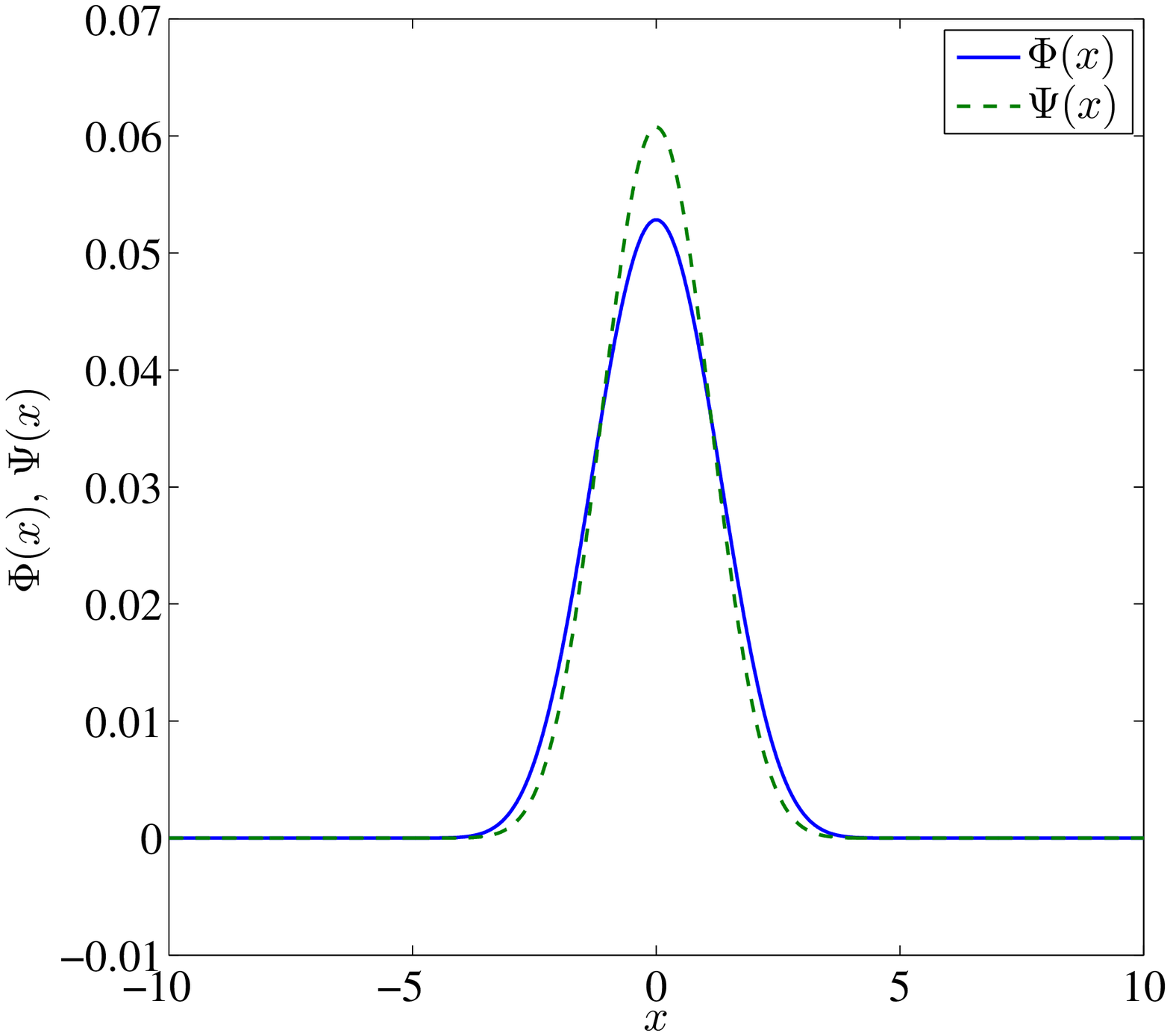}\\
\includegraphics[width=7cm]{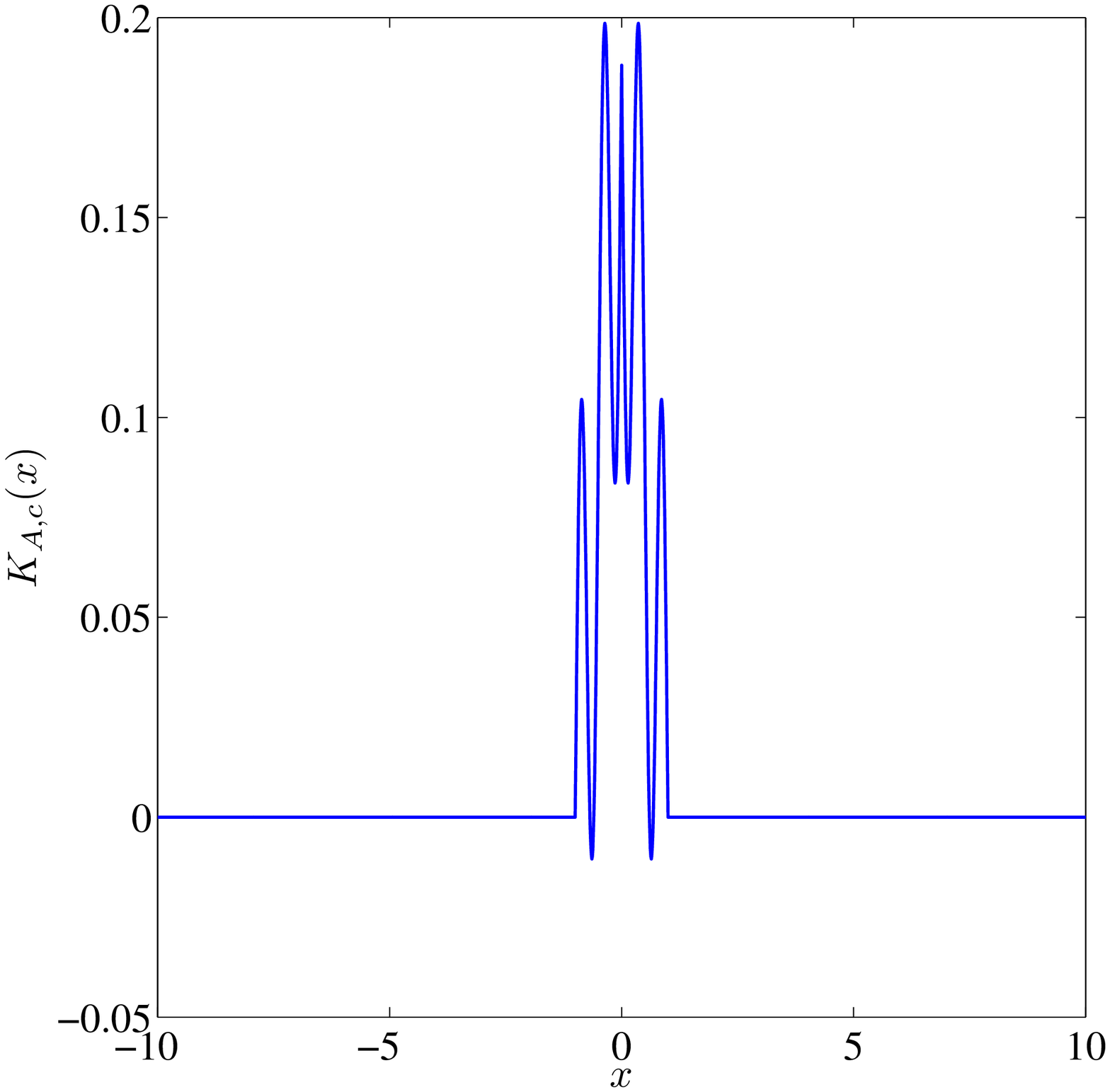}
\includegraphics[width=7cm]{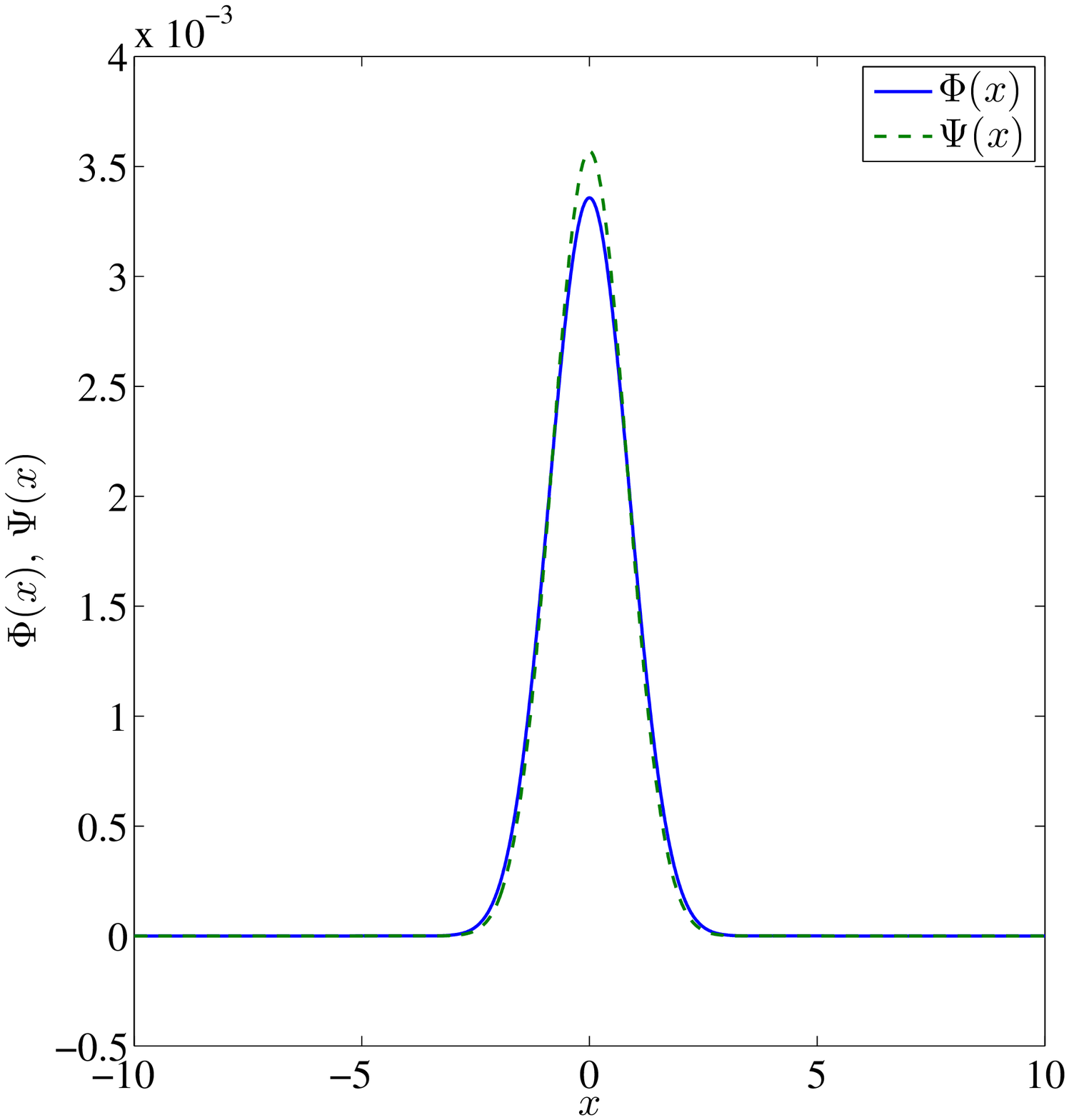}\\
\includegraphics[width=7cm]{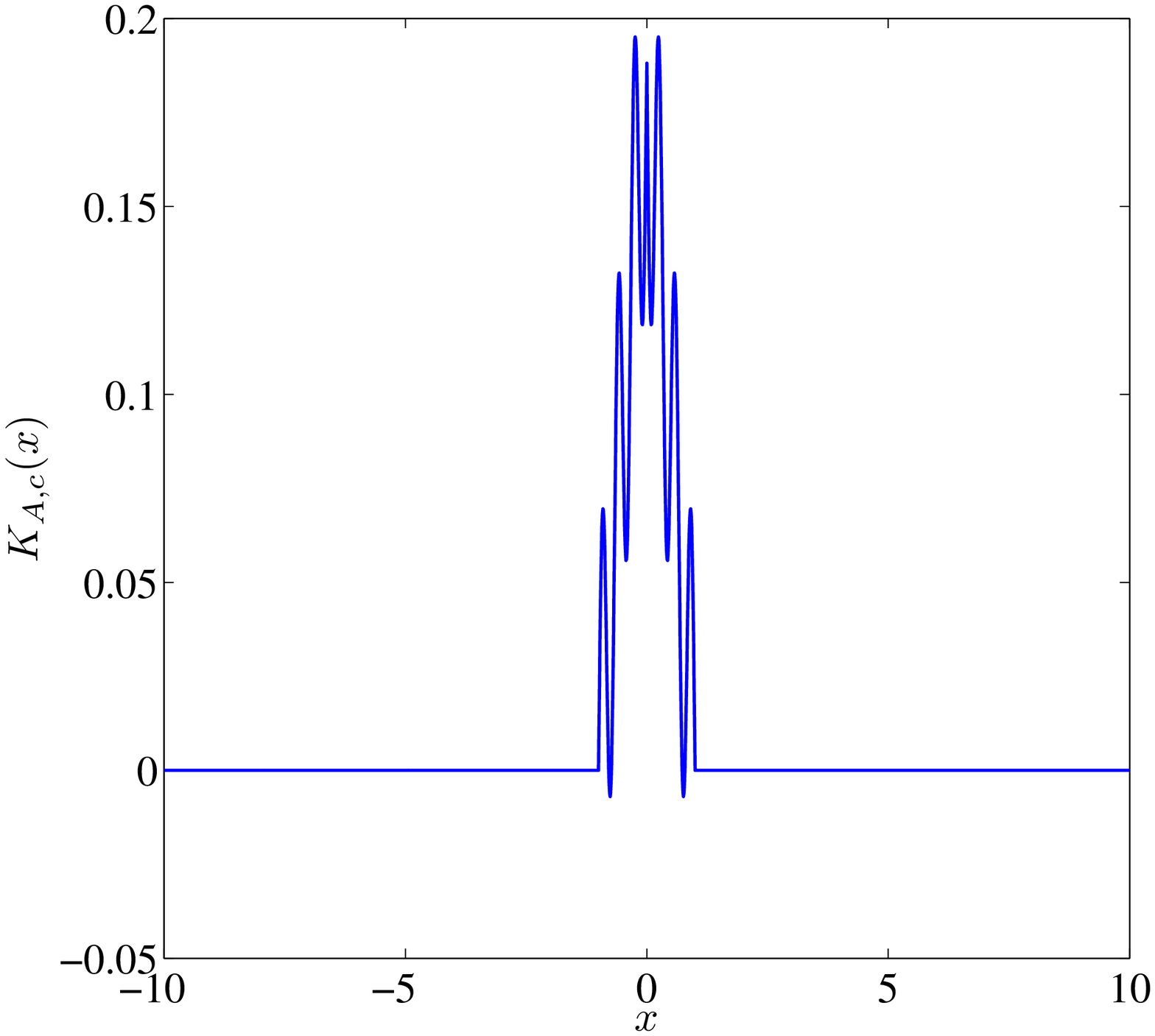}
\includegraphics[width=7cm]{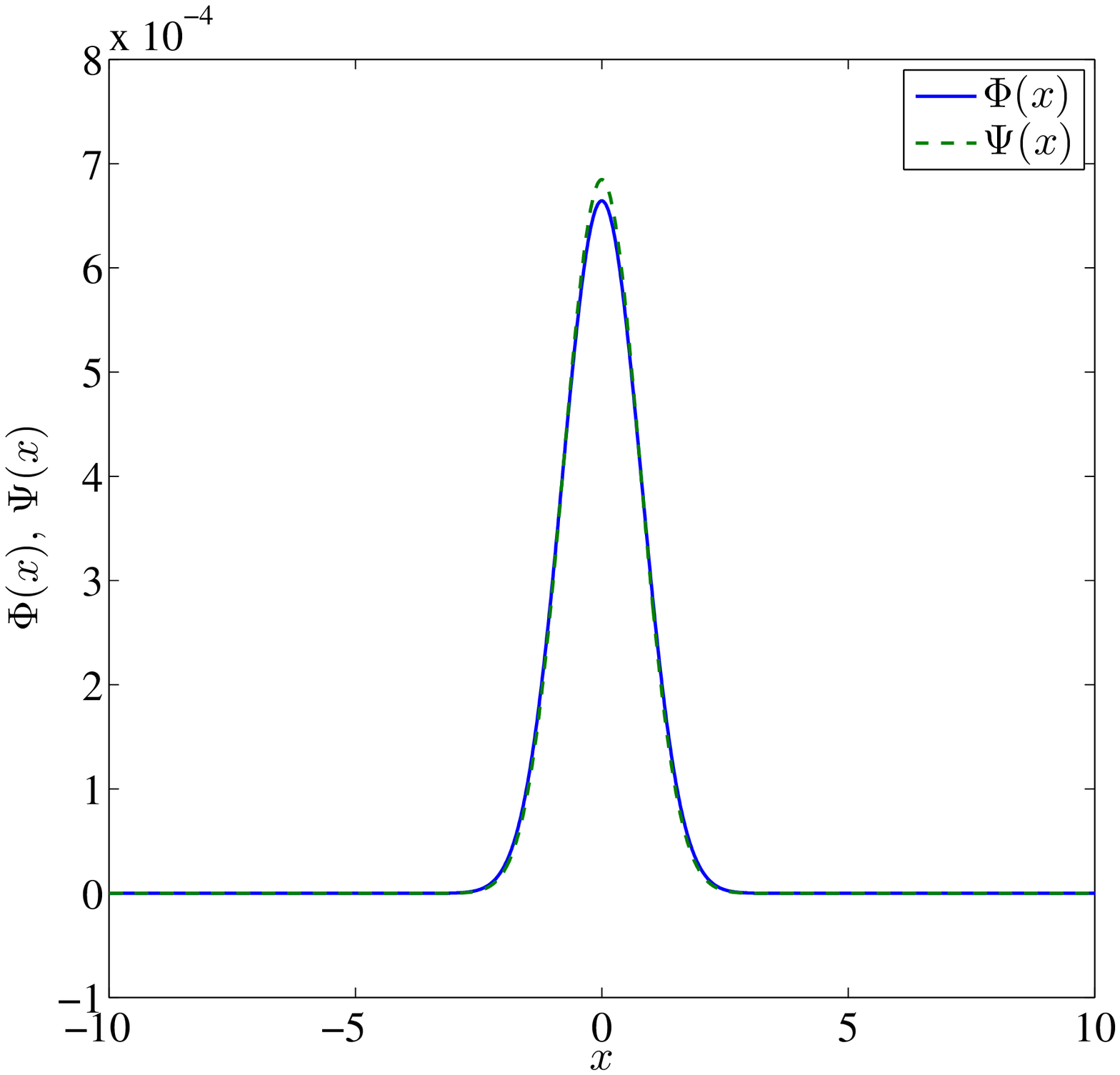}
\end{tabular}
\caption{This figure is similar to the above figures and the only difference is that $A=1.09$ here. The top panels show the kernel of the convolution equation $K_{A,c}(x)$ and corresponding solution $\Phi$ and $\Psi$ with $A=1.09$ and $\mu=2\pi$. In the middle panels and bottom panels, $\mu$ is set as $4\pi$ and $6\pi$, respectively.}
\label{fig3}
\end{figure}

\item  When $1<A<A_{1,n}$, the substantial intervals of negative
values of $K_{A,c}(x)$ lead our numerical scheme to failure of convergence
to any solution for $\Phi$ and $\Psi$. The plot of a typical
example of $K_{A,c}(x)$ within this category
is given in the left panel of Fig.~\ref{fig4}. The dependence
of $A_{1,n}$ on $n$ is shown in the right panel of the figure, as
indicated above.

\begin{figure}[!htbp]
\begin{tabular}{cc}
\includegraphics[width=7cm]{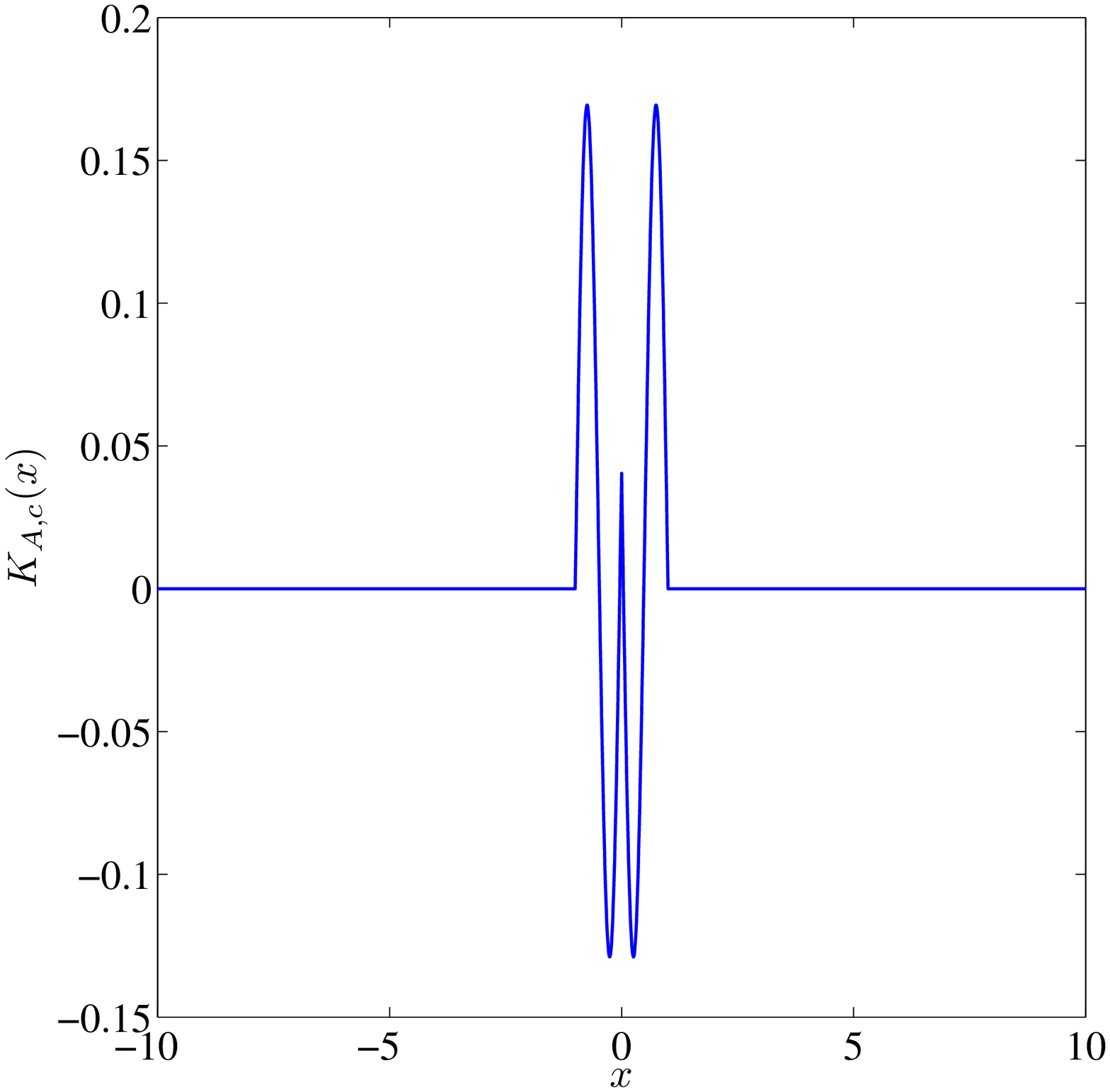}
\includegraphics[width=7cm]{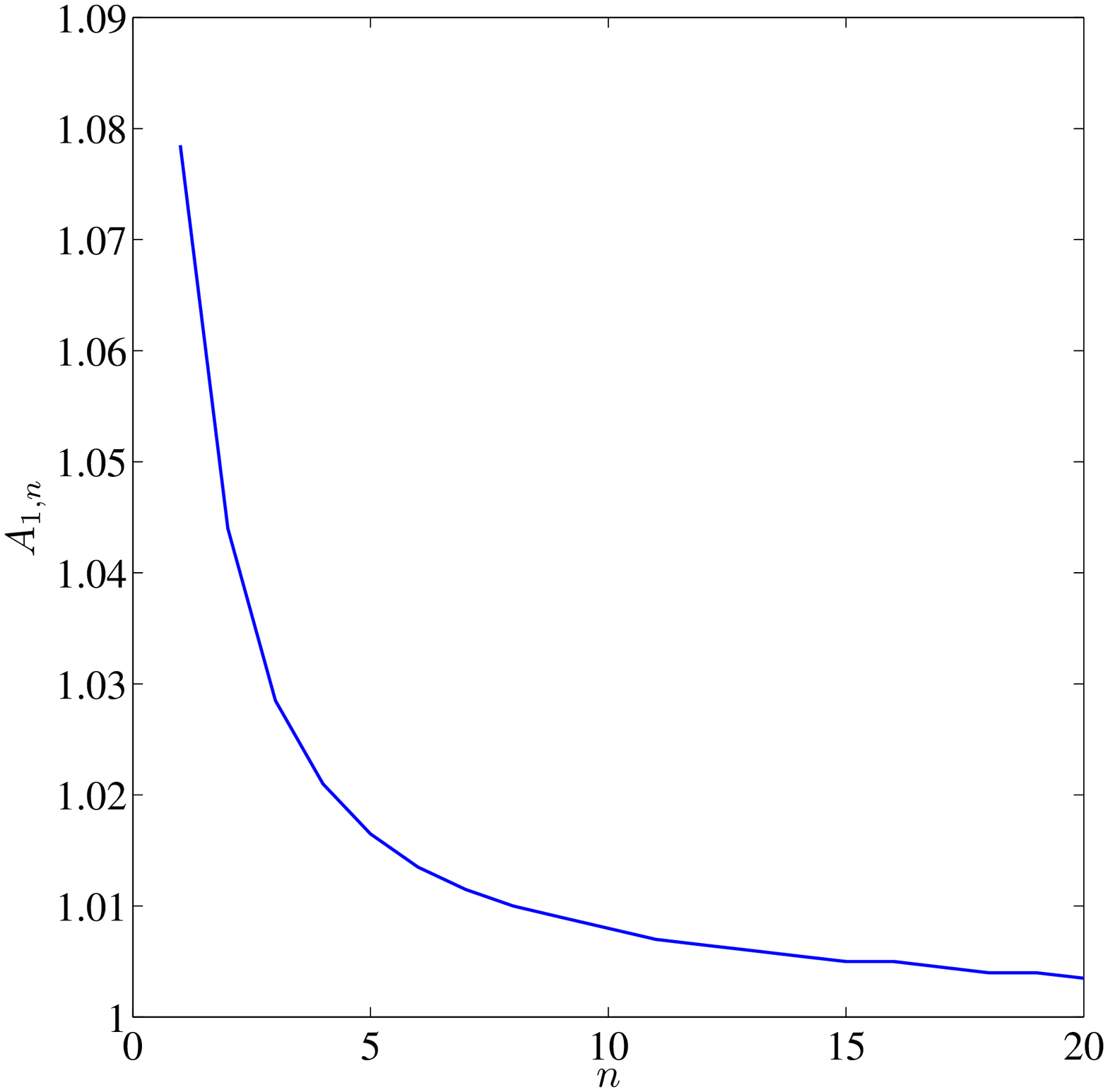}
\end{tabular}
\caption{The left panels shows the kernel of the convolution equation $K_{A,c}(x)$ with $A=1.02$ and $\mu=2\pi$ while how $A_{1,n}$ decreases over $n$ is revealed in the right panel.}
\label{fig4}
\end{figure}

\item Resonance case with $\mu\neq 2\pi n$. In this case,
the conditions for calculating the Fourier transform on $\Phi$ and $\Psi$
fail. By utilizing Fourier series instead (on a finite
computational interval), we obtain that
$K_{A,c}=(A^2-1)\max{(1-|x|,0)}+G(x)$ has non-decaying oscillatory tails on both wings. Moreover, despite the fact that $K_{A,c}$ is neither increasing nor
non-negative on $(-\infty,0)$, we are able to obtain
numerical solutions for $\Phi$ and $\Psi$ from either the convolution
equations (i.e. the integral equations) or the system of advance-delay
differential equations, as shown in Fig.~\ref{fig5}.
It is interesting that the solution also has a bell shape in the
center, but it possesses oscillatory tails on both wings.
The computations of this resonant case and their subtleties (from
a numerical perspective), as well as the different methods
utilized to obtain the solutions are analyzed in detail
elsewhere~\cite{stekev5}. The rigorous analysis which is the main emphasis of
the present work cannot, unfortunately,  presently provide
definitive insights about the latter case. This remains an
important open problem for future study.

\begin{figure}[!htbp]
\begin{tabular}{cc}
\includegraphics[width=7cm]{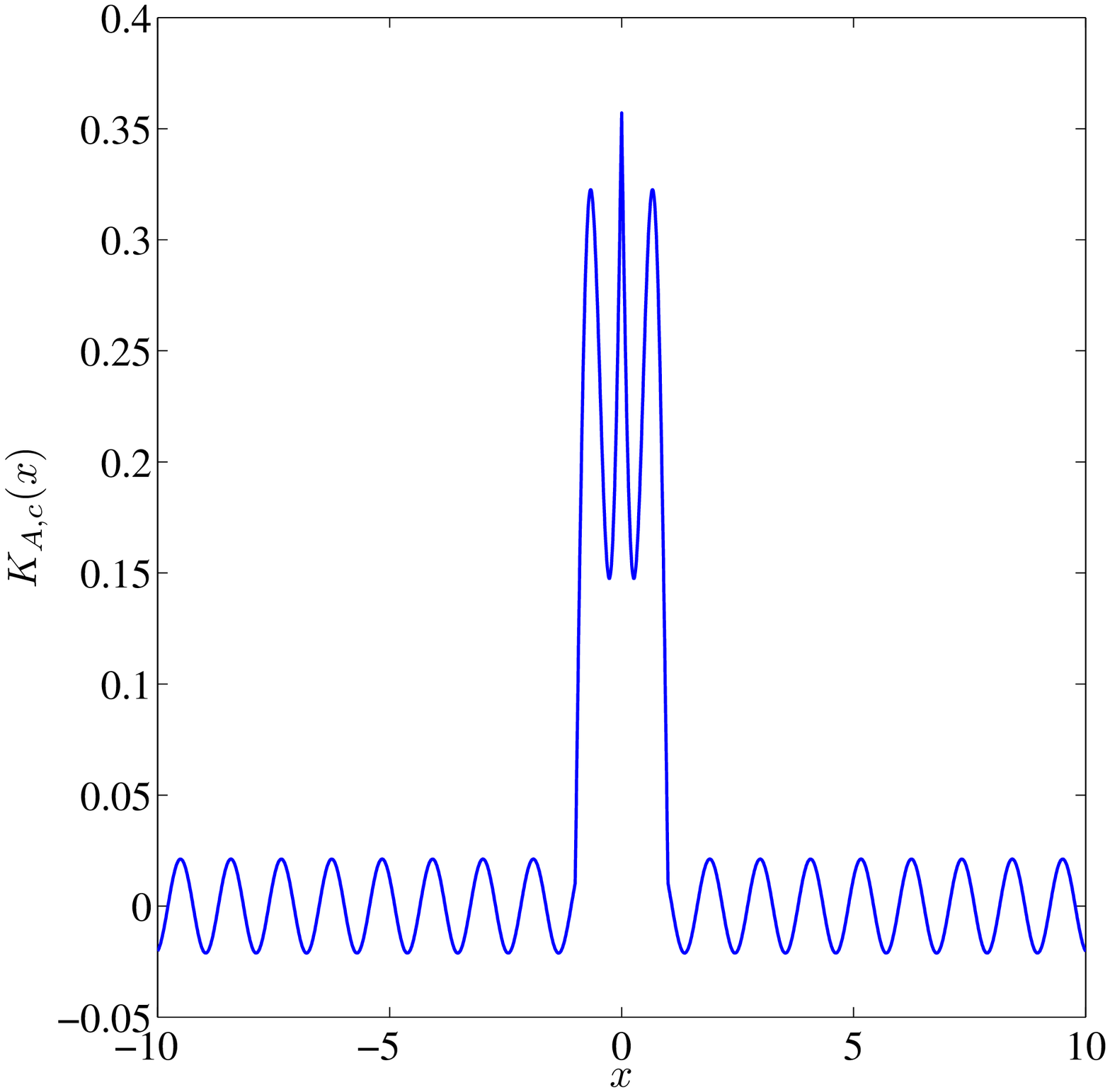}
\includegraphics[width=7cm]{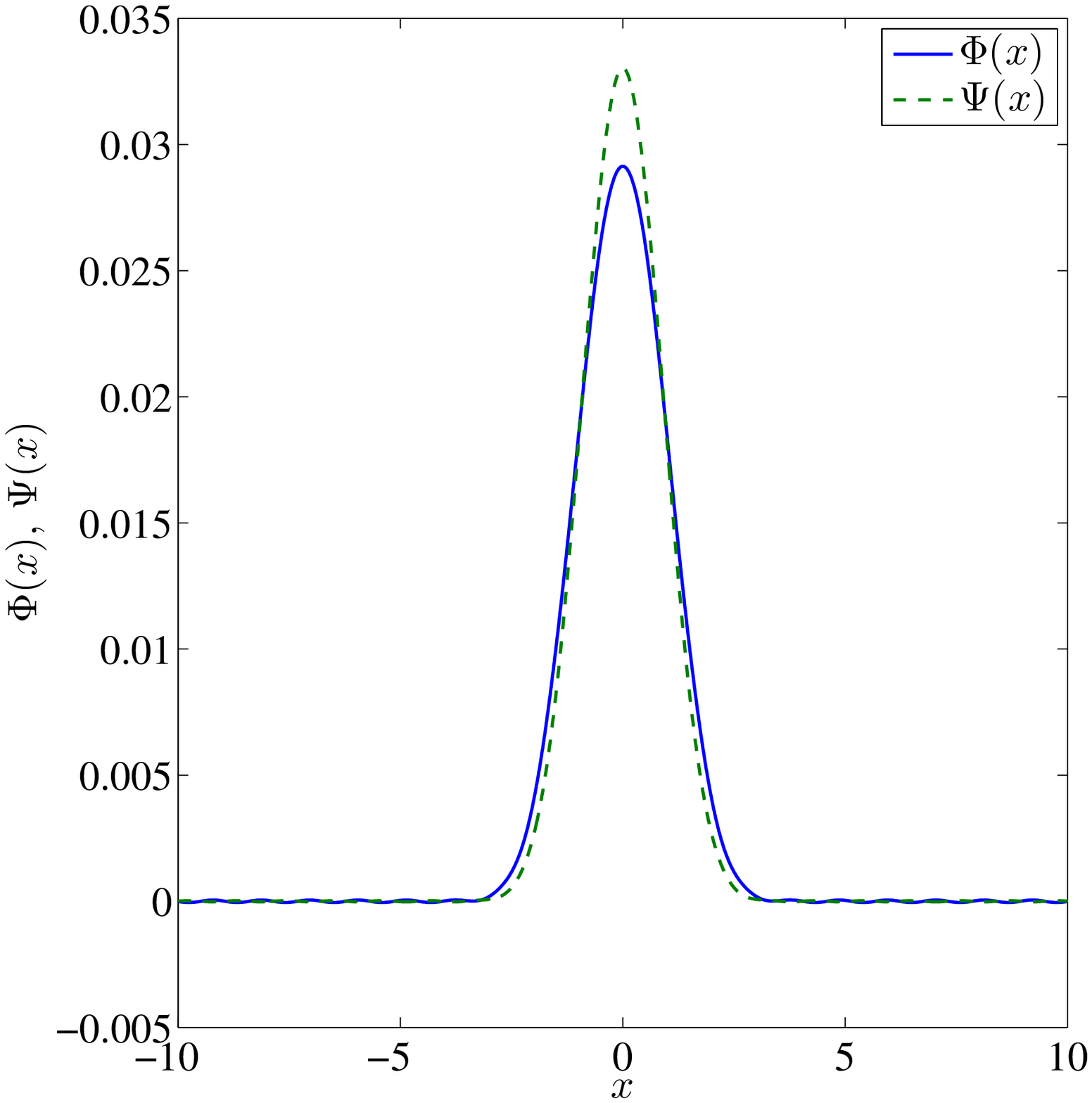}
\end{tabular}
\caption{The left panel shows the kernel of the convolution equation $K_{A,c}(x)$ and the right shows the corresponding solution $\Phi$ and $\Psi$ with $A=1.2$ and $\mu=2\pi-0.5$.}
\label{fig5}
\end{figure}

\end{enumerate}

\section{Conclusions \& Future Challenges}

In the present work, we have explored systems in the form of
Mass in Mass or Mass with Mass dynamical lattices possessing
an internal resonator, studying their traveling waves.
In particular, we could explore the case of so-called anti-resonances
whereby for a particular set of ``quantized'' relations between the
resonator mass and the wave speed (and for suitably
small resonator masses), a bell-shaped traveling wave
could be rigorously proven to exist. Interestingly,
the bell-shaped waves were numerically found to
exist when the anti-resonance condition applies, even
beyond the mass threshold for which our proof holds.
However, a secondary threshold was identified for sufficiently
large resonator masses, beyond which our iterative scheme
for the numerical identification of the waves no longer converged.
Finally, when the anti-resonance condition is not upheld, typical
findings suggest the existence of a wave of non-vanishing tail,
in line also with recent experimental observations~\cite{yang}.

Naturally, many directions of future research emerge from
the present (and recent) considerations of this novel class
of systems. From a rigorous perspective, understanding the
phenomenon of resonances and the formation of traveling
waves with tails would be especially interesting. Equally
interesting from an experimental perspective appears to be
the actual experimental setting whereby an initial excitation
of the chain, by construction, produces such tails which are
arising only on one wing but not the other. Additionally,
as indicated in~\cite{yang}, a distinct experimental possibility
is that of bearing multiple resonators rather than one in
the context of the woodpile configuration. Understanding
how one vs. more such resonators may affect the observed
phenomenology is another important direction for both
theoretical and experimental future studies.


\end{document}